%% file: main.tex
\documentclass[journal,twoside]{IEEEtran}

\usepackage{cite}
\usepackage[pdftex]{graphicx}
\graphicspath{{Figure/}}
\usepackage[cmex10]{amsmath}
\usepackage{amsthm}
\usepackage{amssymb}
\usepackage{algorithmic}
\usepackage{array}
\usepackage{color}
\usepackage{epstopdf}
\usepackage{amsfonts}
\usepackage[acronym,shortcuts]{glossaries}

\usepackage{pgfplots}
\pgfplotsset{compat=1.3}
\usepackage{tikz}							
\usetikzlibrary{shapes}
\usetikzlibrary{spy}
\usetikzlibrary{circuits}
\usetikzlibrary{arrows}

\newcommand{\vect}[1]{\boldsymbol{\mathrm{#1}}}
\newcommand{\mat}[1]{\boldsymbol{\mathrm{#1}}}

\newcommand{\diag}{\mathrm{diag}}

\newcommand{\linewidthmod}{0.98\linewidth}
\newcommand{\linewidthmodd}{0.8\linewidth}

\makeglossaries
\loadglsentries{abbr}

\usepackage{mathtools}

\theoremstyle{definition}
\newtheorem{theorem}{Theorem}
\newtheorem{remark}{Remark}
\newtheorem{proposition}{Proposition}

\newtheorem{corollary}{Corollary}
\newtheorem{heuristic}{Heuristic}
\newtheorem{conjecture}{Conjecture}

\ifCLASSOPTIONcompsoc
  \usepackage[caption=false,font=normalsize,labelfont=sf,textfont=sf]{subfig}
\else
  \usepackage[caption=false,font=footnotesize]{subfig}
\fi
\usepackage{dblfloatfix}

\begin{document}

\title{\huge The Z3RO Family of Precoders Cancelling Nonlinear Power Amplification Distortion in Large Array Systems}

\author{Fran\c{c}ois Rottenberg, Gilles Callebaut, Liesbet Van der Perre
	\thanks{Fran\c{c}ois Rottenberg, Gilles Callebaut and Liesbet Van der Perre are with ESAT-DRAMCO, Ghent Technology Campus, KU Leuven, 9000 Ghent, Belgium (e-mail: francois.rottenberg@kuleuven.be).}
}

\markboth{DRAFT}%
{}
%



\maketitle

\begin{abstract}
	Large array systems use a massive number of antenna elements and clever precoder designs to achieve an array gain at the user location. These precoders require linear front-ends, and more specifically linear power amplifiers (PAs), to avoid distortion. This reduces the energy efficiency since PAs are most efficient close to saturation, where they generate most nonlinear distortion. Moreover, the use of conventional precoders can induce a coherent combining of distortion at the user locations, degrading the signal quality. In this work, novel linear precoders, simple to compute and to implement, are proposed that allow working close to saturation, while cancelling the third-order nonlinearity of the PA without prior knowledge of the signal statistics and PA model. Their design consists in saturating a single or a few antennas on purpose together with an negative gain with respect to all other antennas to compensate for the overall nonlinear distortion at the user location. The performance gains of the designs are significant for PAs working close to saturation, as compared to maximum ratio transmission (MRT) precoding and perfect per-antenna digital pre-distortion (DPD) compensation.
\end{abstract}


%
\IEEEpeerreviewmaketitle

\input{Section/Introduction.tex}

\input{Section/System_Model.tex}

\input{Section/MRT_limitations.tex}

\input{Section/Z3RO_precoder.tex}

\input{Section/Simulation_Results.tex}

\input{Section/Conclusion.tex}

\input{Section/Appendix.tex}

%
%

\ifCLASSOPTIONcaptionsoff
  \newpage
\fi



\bibliographystyle{IEEEtran}

\bibliography{IEEEabrv,refs}
\end{document}

%% file: abbr.tex
\newacronym[plural=BSs,firstplural=base stations (BSs)]{bs}{BS}{base station}
\newacronym[plural=PAs,firstplural=power amplifiers (PAs)]{pa}{PA}{power amplifier}
\newacronym{sdr}{SDR}{signal-to-distortion ratio}
\newacronym{snr}{SNR}{signal-to-noise ratio}
\newacronym{sndr}{SNDR}{signal-to-noise-and-distortion ratio}
\newacronym{snidr}{SNIDR}{signal-to-noise-and-interference-and-distortion ratio}
\newacronym{los}{LOS}{line-of-sight}
\newacronym{z3ro}{Z3RO}{zero third-order distortion}
\newacronym{mimo}{MIMO}{multiple input multiple output}
\newacronym{mrt}{MRT}{maximum ratio transmission}
\newacronym{dpd}{DPD}{digital pre-distortion}

\newacronym{ula}{ULA}{uniform linear array}

\newacronym{psd}{PSD}{power spectral density}
\newacronym{aclr}{ACLR}{adjacent channel leakage ratio}
\newacronym{cpu}{CPU}{central processing unit}
\newacronym{iid}{i.i.d.}{independently and identically distributed}
\newacronym{ue}{UE}{user equipment}
\newacronym{nlos}{NLoS}{non-line-of-sight}
\newacronym{hpbm}{HPBM}{half power beam width}
\newacronym{rts}{RTS}{ray tracing simulator}
\newacronym{ag}{AG}{array gain}
\newacronym{arp}{ARP}{Antenna Reference Point}
\newacronym{ecdf}{eCDF}{Empirical Cumulative Distribution Function}
\newacronym{evm}{EVM}{Error Vector Magnitude}
\newacronym{fft}{FFT}{fast Fourier transform}
\newacronym{ib}{IB}{in-band}
\newacronym{im}{IM}{intermodulation}
\newacronym{mf}{MF}{matched filtering}
\newacronym{mr}{MR}{maximum ratio}
\newacronym{ofdm}{OFDM}{orthogonal frequency division multiplexing}
\newacronym{oob}{OOB}{out-of-band}
\newacronym{papr}{PAPR}{peak-to-average power ratio}
\newacronym{rx}{RX}{Receiver}
\newacronym{trp}{TRP}{Transmission Reception Point}
\newacronym{tx}{TX}{transmitter}
\newacronym{zf}{ZF}{zero forcing}
\newacronym{dl}{DL}{downlink}
\newacronym{rzf}{RZF}{regularized zero forcing}
\newacronym{slc}{SLC}{spatial leakage suppression}
\newacronym{kpi}{KPI}{key performance indicator}
\newacronym{awgn}{AWGN}{additive white gaussian noise}
\newacronym{qam}{QAM}{quadrature amplitude modulation}
\newacronym{amam}{AM/AM}{amplitude modulation to amplitude modulation}
\newacronym{ampm}{AM/PM}{amplitude modulation to phase modulation}
\newacronym{zmcscg}{ZMCSCG}{zero mean circularly symmetric complex Gaussian}

%% file: Section/Introduction.tex
\section{Introduction}

\subsection{Problem Statement}
Energy efficiency improvements and carbon footprint reduction are main concerns in our society and priorities of governments, such as expressed in Europe's Green Deal~\cite{EUGreenDeal}. Inefficient operation of the \gls{pa} has been identified as the main contributor to overall energy budget of wireless networks. Indeed, typical spectral efficient transmission schemes exhibit a high dynamic range, calling for a significant back-off in the amplifier operation. This brings about an inevitable trade-off between linearity and efficiency~\cite{Lavrador}. On the one hand, the linearity of the \gls{pa} directly impacts the signal quality. On the other hand, the \gls{pa} has a maximal efficiency when it is operating is close to saturation, where its characteristic is nonlinear and creates distortion. The latter can degrade the quality of the communication link, and moreover generates \gls{oob} distortions. This results in trade-offs between system capacity and power consumption.

Recent works have studied how the transmission with large antenna arrays can impact the trade-off between linearity and energy efficiency~\cite{Muneer,Fager,Moghadam2018,Teodoro2019}. In massive \gls{mimo} systems, authors have shown that the \gls{pa} nonlinear distortion is not always uniformly radiated~\cite{Larsson2018,Mollen2018TWC}. Measurements performed with an actual massive \gls{mimo} testbed also confirmed that the distortion terms can appear as correlated noise and hence can coherently combine~\cite{Zou2015}. In many situations, nonlinear distortion coherently combines in the direction(s) of the intended user(s) and more specifically at its location. This is in particular pronounced for a scenario with a single or a few users and dominant propagation direction(s), \textit{e.g.}, a strong \gls{los} component. The analysis has been extended to deployments with distributed arrays~\cite{rottenberg2021spatial}. When many intermodulation beams~\cite{Hemmi2002} appear, with increasing number of users and multi-path components in the channels, the transmit signals become ‘isotropically’ radiated~\cite{Mollen2018TWC}. More importantly, novel precoders designed with adequate performance criteria, as introduced in this paper, can avoid this coherent combining altogether.

\subsection{State-of-the-Art}

Allowing the \glspl{pa} to work close to saturation while preventing degradation of the link quality has been a vast area of research in the literature. One solution is to use a low \gls{papr} waveform or precoding technique, such as the constant-envelope precoder~\cite{Khan2013}. However, their adoption has been hindered by the associated high digital processing complexity. Another technique is to use reserved tones of an \gls{ofdm} system to design a peak-cancelling signal that lowers the \gls{papr} of the signal fed to the \glspl{pa}~\cite{Krongold2004}. Related challenges are the loss of spectral efficiency and the complexity of optimizing the data-dependent cancelling signal. Another solution is to use digital pre-distortion compensation techniques: the signal that is at the input of the \gls{pa} is pre-distorted so as to compensate for distortions created by the \gls{pa}. This enables higher PA efficiencies~\cite{cripps2006rf}. It has also been studied for massive \gls{mimo}~\cite{Tervo21, Tarver21}. These techniques are however data-dependent and often require a feedback loop. Applying these techniques requires a certain complexity burden and causes scalability problems, especially in a massive \gls{mimo} setting, where it has to be implemented for each \gls{pa}. Moreover, even a perfect digital pre-distortion only linearizes input signals with an amplitude lower than the saturation level of the PA (weakly nonlinear effects). Higher fluctuations are clipped due to saturation, resulting in nonlinear distortion (strongly nonlinear effects). This phenomenon becomes more likely as the back-off power is reduced, especially for high \gls{papr} signals such as multicarrier and massive \gls{mimo} systems. A neural network approach for \gls{dpd} has been considered in~\cite{Tarver2021}, where the DPD is applied before the precoder and scales with the number of users rather than the number of antennas. Another suggested solution is to use or create differences between the \gls{pa} characteristics to induce noncoherent combining of distortion at the user location~\cite{Anttila2019}.

In contrast to previous techniques, this work focuses on the design of novel linear precoders, which are simple to implement and to compute. A recent work~\cite{aghdam2020distortion} has followed a similar approach. However, no simple closed-form solution for the precoder could be found. An optimization approach is used as the algorithm is iterative, requires computing gradients of large matrices and projections. Moreover, due to the non-concavity of the problem, the authors need to perform the optimization several times with different initialization points. Another recent work~\cite{Kolomvakis2022} has considered the design of linear precoders to null nonlinear distortion towards victim users. The authors used a sub-optimal approach based on constant envelope precoding coefficients. In the work \cite{Zayani2019}, the symbols prior to the PAs are optimized with an iterative algorithm to minimize the mean squared error at the receive side. An iterative algorithm based on zero forcing precoding was considered in~\cite{Iofedov2015} to cancel nonlinear distortion and avoid user interference in an \gls{ofdm} system. The algorithm can also keep out-of-band radiation low so as to preserve spectral purity. As opposed to these previous works, the approach in this work is purely analytical. No iterative algorithms are required and we propose precoders in closed-form. Moreover, no constraint on the envelope of the precoder is imposed, leading to a more performant design.




\subsection{Contributions}

In this paper, we present an original contribution: the \gls{z3ro} family of linear precoders which have a low complexity and maximize the received \gls{snr} while completely canceling the third-order nonlinear distortion at the user location. More specifically, the structure of our paper and our contributions are structured as follows. Section~\ref{section:system_model} describes the system model. Sections~\ref{section:MRT_limitations} presents the limitation of the \gls{mrt} precoder, \textit{i.e.}, MRT is only optimal in the linear regime. As the PAs enter saturation, \gls{mrt} is not optimal due to nonlinear distortion. Section~\ref{section:Z3RO_precoder} presents the Z3RO family of linear precoders. We formulate an optimization problem to design the linear precoder that maximizes the \gls{snr} at the receiver while completely cancelling the third-order distortion at the user location. The problem is not concave. Still, we solve it and fully characterize its maxima. Each maximum requires to perform a line search. In the pure \gls{los} case, it is shown that all maxima are globally optimal and have a simple closed-form solution. For the general channel case, the so-called Z3RO precoder is proposed based on a heuristic design, which has a simple closed-form and is shown to provide good performance. The Z3RO precoder was initially proposed in~\cite{rottenberg2021z3ro}. A measurement-based performance evaluation was performed in~\cite{feys2022measurementbased}. The basic principle of the proposed precoders is to use a set of saturated\footnote{By ``saturated antennas", we refer to antennas with a relatively higher absolute channel gain than the remaining antennas.} antennas with a negative gain as compared to all other antennas so as to compensate their overall distortion. It is shown that the resulting array gain penalty becomes negligible in the large antenna case, and especially in the saturation regime, where the \gls{sndr} is limited by the \gls{sdr} and not by the \gls{snr}. It is also shown that while using a single saturated antenna is optimal in terms of \gls{snr}, using more induces a more spatially focused distortion pattern, which is useful regarding unintended locations and \gls{oob} radiation. Section~\ref{section:Simulation_results} presents simulation results. Finally, Section~\ref{section:conclusion} concludes the paper.

{\textbf{Notations}}: Vectors and matrices are denoted by bold lowercase $\vect{a}$ and uppercase letters $\mat{A}$, respectively (resp.). Superscripts $^*$ and $^T$ stand for conjugate and transpose operators. The symbols $\mathbb{E}$ denotes the expectation. $\jmath$ is the imaginary unit. $\mat{I}_N$ denotes the identity matrix of order $N$. 
The notation $\diag(\vect{a})$ refers to a diagonal matrix whose $k-$th diagonal entry is equal to the $k-$th entry of vector $\vect{a}$. The notation $\vect{g}^a$ element-wise takes the $a$-th power $\vect{g}$.




%% file: Section/System_Model.tex
\section{System Model}
\label{section:system_model}

\subsection{Signal Model}


We consider a large array-based system with a single user and single \gls{bs} equipped with $M$ antennas.  The complex symbol intended for the user is denoted by $s$, with zero mean and variance $p=\mathbb{E}(|s|^2)$. The signal $s$ is precoded at transmit antenna $m$ using a precoder coefficient $w_{m}$. The complex baseband representation of the signal before the \gls{pa} of the corresponding antenna is denoted by $x_{m}$ and is given by $x_{m}= w_{m} s$.


\subsection{Power Amplifier Model}

In the following, all \glspl{pa} are assumed memoryless and have the same transfer function. Their identical nature can be seen as a worst case in terms of coherent combining of distortion~\cite{Anttila2019}. For the sake of clarity and without loss of generality, the linear gain of the \gls{pa} is set to one. We only consider the third order nonlinear distortion of the \gls{pa}. This approximation regime is valid as the \gls{pa} enters saturation regime, which creates nonlinear distortion but not enough for higher order terms to provide a significant contribution. Under these assumptions, the \gls{pa} output of antenna $m$ can be written as
\begin{align}
	y_{m}= x_{m} +a_3 x_{m}  |x_{m}|^{2}, \label{eq:y_m_gen}
\end{align}
where the coefficient $a_3$ characterizes the nonlinear characteristic of the \gls{pa}, including both amplitude-to-amplitude modulation (AM/AM) and amplitude-to-phase modulation (AM/PM).

In the simulation section, a more general PA model, not limited to the third order, will be considered for validation purposes.

\subsection{Channel Model}\label{subsection:channel_model}

The complex channel gain from antenna $m$ to the user is denoted by $h_{m}$. 
The received signal is given by
\begin{align}
	r&=\sum_{m=0}^{M-1} h_{m} y_m + v, \label{eq:r_SU}
\end{align}
where $v$ is zero mean circularly symmetric complex Gaussian noise with variance $\sigma_v^2$. In the following, at some places, a pure \gls{los} channel will be considered. In this particular case, $h_{m}$ can be written as
\begin{align}
	h_{m}= \sqrt{\beta} e^{-\jmath \phi_{m}}. \label{eq:channel_model_LOS}
\end{align}
The real positive coefficient $\beta$ models the path loss. The difference of propagation distance between each of the antennas and the user results in an antenna-dependent phase shift $\phi_{m}$, which can be directly related to the antenna location and the angular direction of the user. For a \gls{ula} and a narrowband system, the phase shift is given by $\phi_{m}=m\frac{2\pi}{\lambda_c} d\cos(\theta)$, where $\lambda_c$ is the carrier wavelength, $d$ the inter-antenna spacing and $\theta$ is the user angle. 

The radiation pattern in an arbitrary direction $\tilde{\theta}$ can be computed as
\begin{align}
	P(\tilde{\theta})&=\mathbb{E}\left(\left|\sum_{m=0}^{M-1}y_me^{-\jmath\tilde{\phi}_m}\right|^2\right)
	,\label{eq:radiation_pattern}
\end{align}
where $\tilde{\phi}_m=m \frac{2\pi}{\lambda_c} d\cos(\tilde{\theta})$. Defining the total transmit power $P_T=\int_{-\pi}^{\pi}P(\tilde{\theta})d\tilde{\theta}$, the array directivity is $D(\tilde{\theta})=\frac{P(\tilde{\theta})}{P_T/2\pi}$, \textit{i.e.}, $P(\tilde{\theta})$ is normalized with respect to an isotropic radiator. In the following, the radiation and directivity pattern of the linear and nonlinear parts of the amplified signal are depicted in several figures. They are respectively obtained by replacing $y_m$ in (\ref{eq:radiation_pattern}) by $x_m$ and $a_3 x_m |x_m|^2$.

%% file: Section/MRT_limitations.tex
\section{Maximum Ratio Transmission Precoder}
\label{section:MRT_limitations}

\begin{figure*}[t!]
	\centering 
	\subfloat[MRT, $M=32$.
	]{
		\resizebox{0.3\linewidth}{!}{%
			{\includegraphics[clip, trim=2.5cm 2.6cm 2.5cm 3cm, scale=1]{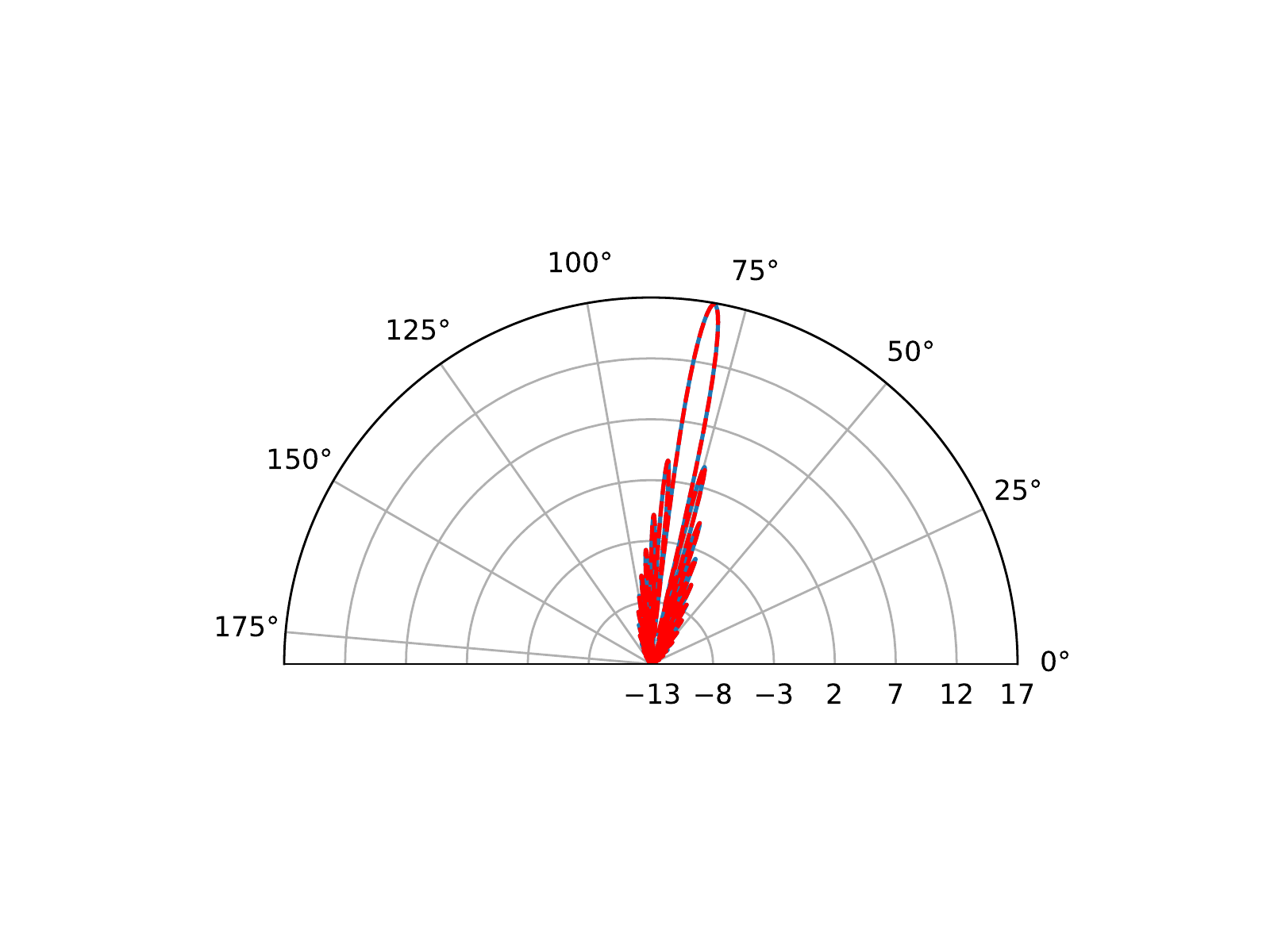}} 
		}
	} 
	\subfloat[\gls{z3ro}, $M=2$.]{
		\resizebox{0.3\linewidth}{!}{%
			{\includegraphics[clip, trim=2.5cm 2.6cm 2.5cm 3cm, scale=1]{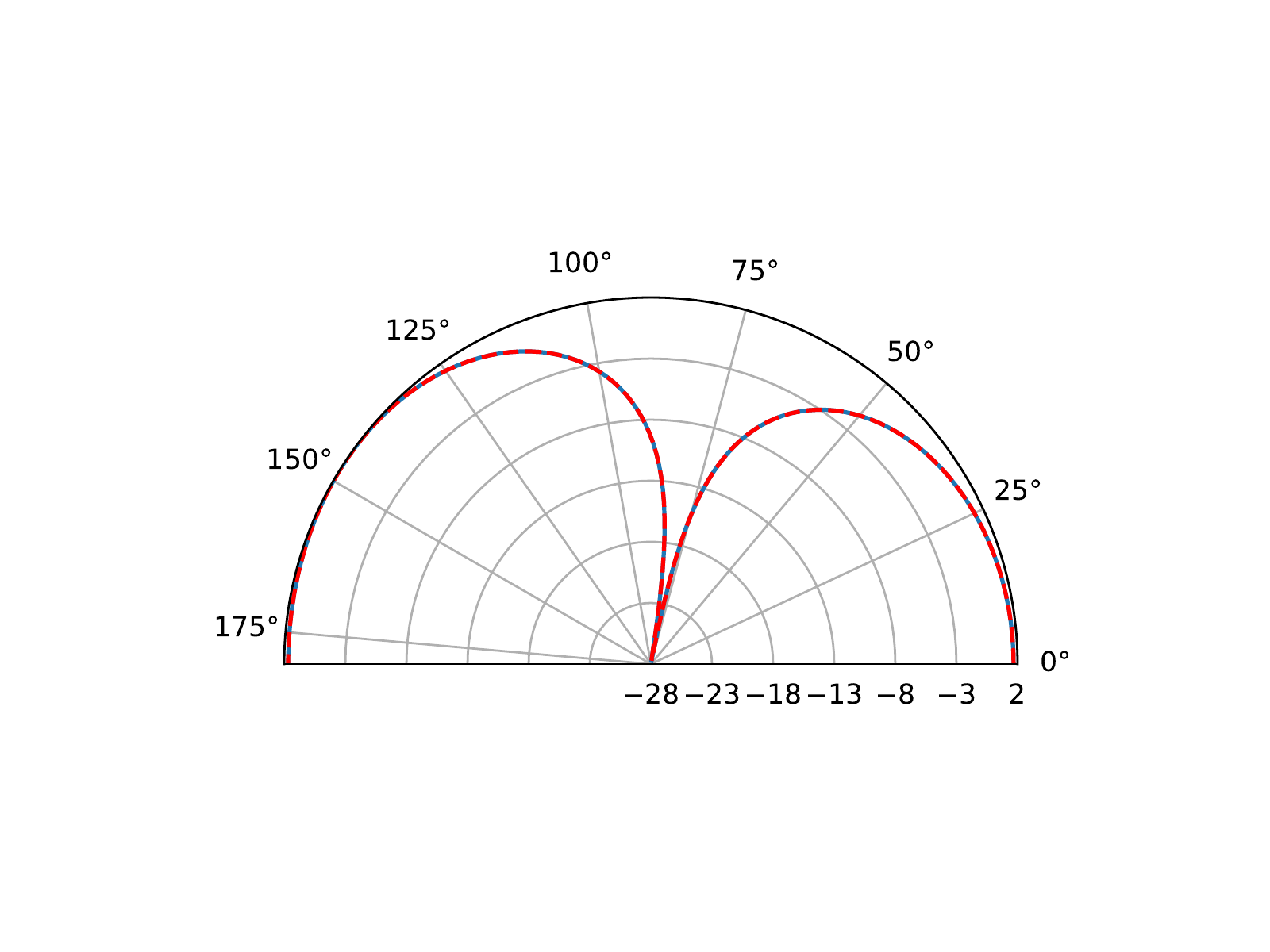}} 
		}
	}\\\vspace{-1em}
	\subfloat[\gls{z3ro}, $M=8$.]{
		\resizebox{0.3\linewidth}{!}{%
			{\includegraphics[clip, trim=2.5cm 2.6cm 2.5cm 3cm, scale=1]{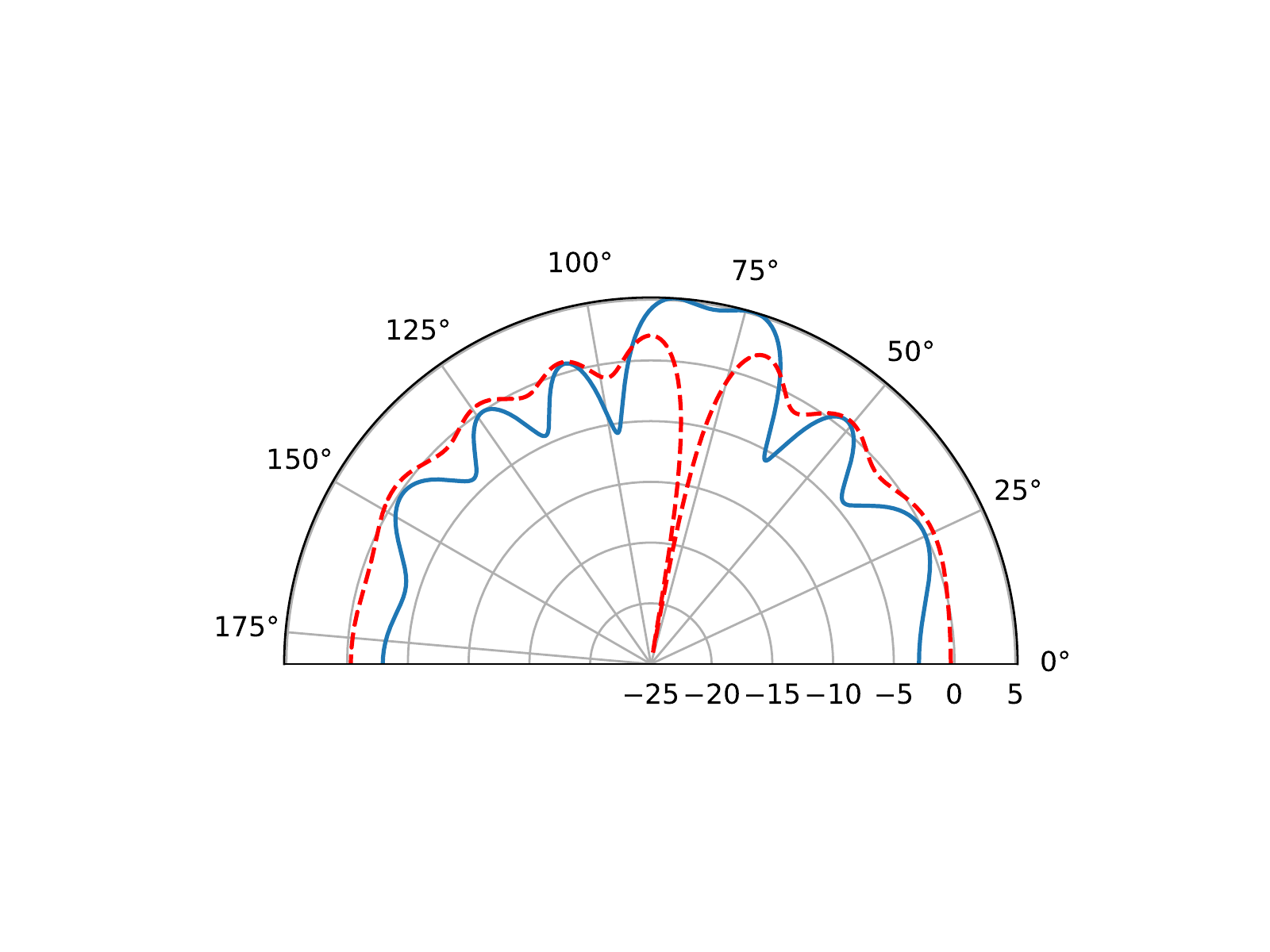}} 
		}
	}
	\subfloat[\gls{z3ro}, $M=32$.]{
		\resizebox{0.3\linewidth}{!}{%
			{\includegraphics[clip, trim=2.5cm 2.6cm 2.5cm 3cm, scale=1]{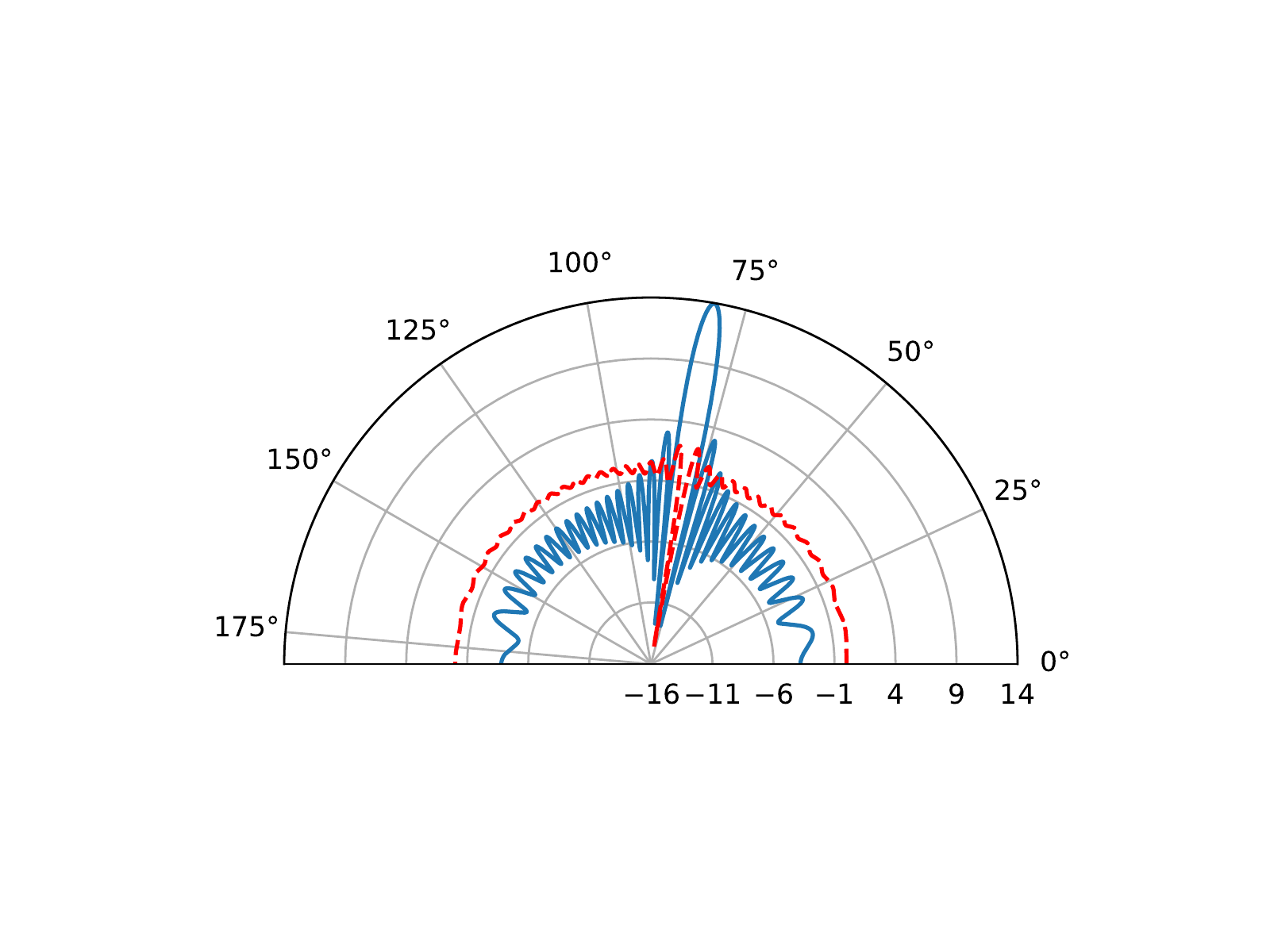}} 
		}
	}
	\caption{Directivity pattern [dB] of the signal (continuous blue) and third-order distortion (dashed red) for a pure \gls{los} channel and half-wavelength \gls{ula}. User angle is $\theta=80^{\circ}$ and $M_s=1$ saturated antenna. For \gls{mrt}, distortion coherently combines in the user direction while it is null for \gls{z3ro}. The array gain penalty is infinite for $M=2$ while it vanishes as $M$ grows large.}
	\label{fig:directivity_patterns} 
	\vspace{-1em}
\end{figure*}

The well known \gls{mrt} precoder is obtained by maximizing the received \gls{snr} under a transmit power constraint, disregarding the nonlinear distortion terms at the output of the \gls{pa}. Its expression and the SNR at the user are given by \cite{paulraj2003introduction}
\begin{align*}
	w_m^{\mathrm{MRT}}&= \frac{h_m^*}{\sqrt{\sum_{m'=0}^{M-1}|h_{m'}|^2}},\ \text{SNR}^{\mathrm{MRT}}=  \frac{p\sum_{m'=0}^{M-1}|h_{m'}|^2}{\sigma_v^2}.
\end{align*}
The \gls{mrt} achieves an array gain of a factor $M$. The \gls{mrt} precoder is optimal as long as the \gls{pa} works in its linear regime. As $p$ increases, nonlinear terms will be amplified and distortion becomes non-negligible. The \gls{pa} output $y_m$ given in (\ref{eq:y_m_gen}) can be evaluated for $x_m=w_m^{\mathrm{MRT}} s$
\begin{align*}
	y_{m}&= s \alpha h_{m}^*   + a_3 s  |s|^{2} \alpha ^3 h_{m}^*  |h_{m}|^{2},
\end{align*}
where $\alpha= {1}/{\sqrt{\sum_{m'=0}^{M-1}|h_{m'}|^2}}$. The received signal (\ref{eq:r_SU}) becomes
\begin{align*}
	r= s  \alpha \sum_{m=0}^{M-1}|h_{m}|^{2} +a_3s  |s|^{2}  \alpha ^3 \sum_{m=0}^{M-1}|h_{m}|^{4} + v.
\end{align*}
This expression shows that the channel coherently combines both the linear term and the nonlinear term, \textit{i.e.}, their phases are matched. As a result, distortion coherently adds up at the user location and becomes the limiting factor at high power. To illustrate this, let us consider the pure \gls{los} channel introduced in (\ref{eq:channel_model_LOS}). Then,
\begin{align*}
	y_{m}&= \frac{e^{\jmath \phi_m}}{\sqrt{M}} (s    + a_3 s  |s|^{2}) \\
	r&= \sqrt{\beta M} (s+ a_3s  |s|^{2})+ v,
\end{align*}
where it is clear that the array gain $M$ affects both linear and nonlinear terms. Moreover, one can see that both the linear and nonlinear terms are beamformed in the same direction, as they are both affected by the term $e^{\jmath \phi_m}$. An example of the directivity pattern of linear/nonlinear terms for a \gls{ula} is shown in Fig.~\ref{fig:directivity_patterns}~(a).

%% file: Section/Z3RO_precoder.tex
\section{Zero Third-Order (Z3RO) Family of Precoders}
\label{section:Z3RO_precoder}

The \gls{mrt} precoder induces a coherent combining of distortion at the user location, as demonstrated in the previous section. As $p$ increases, the \glspl{pa} will become more saturated and the user performance will be limited by its \acrlong{sdr} (SDR). This section presents the family of \gls{z3ro} precoders, which are designed to maximize the \gls{snr} at the user location while cancelling the combining of third order distortion.

Inserting (\ref{eq:y_m_gen}) in (\ref{eq:r_SU}), the received signal at the user location for a general linear precoder $w_m$ is
\begin{align*}
	r &= s \sum_{m=0}^{M-1} h_m w_m +a_3 s |s|^2 \sum_{m=0}^{M-1} h_m w_m |w_m|^2 + v.
\end{align*}
The distortion term can be forced to zero by ensuring that
\begin{align}
	\sum_{m=0}^{M-1} h_m w_m |w_m|^2&=0. \label{eq:zero_distortion_constraint}
\end{align}
This constraint does not depend on the transmit symbol $s$ and the \gls{pa} parameter $a_3$, which makes it practical to implement. A similar condition was obtained in \cite{aghdam2020distortion}. However, the authors made the pessimistic conclusion that considering this constraint leads to a considerable reduction of array gain. Indeed, take the two antenna case $M=2$ and a \gls{los} channel $h_m=\sqrt{\beta} e^{-\jmath \phi_m}$. If the user angle is coming from broadside, it implies that $\phi_0=\phi_1=0$ and the constraint (\ref{eq:zero_distortion_constraint}) implies that
\begin{align*}
	w_0|w_0|^2&=-w_1|w_1|^2 \Leftrightarrow w_0=-w_1,
\end{align*}
which leads to a zero array gain, \textit{i.e.}, $|w_0+w_1|^2=0$. The same result occurs for any user angle, as depicted in Fig.~\ref{fig:directivity_patterns}~(b). However, it is shown further that, with the proposed designs, as the number of antennas $M$ grows large, the loss in array gain becomes negligible, as depicted in Fig.~\ref{fig:directivity_patterns}~(c) and (d). The precoder optimization problem can be formulated as
\begin{align}
	\max_{w_0,...,w_{M-1}} \text{SNR}=\frac{p}{\sigma_v^2}  \left|\sum_{m=0}^{M-1} h_m w_m\right|^2, \label{eq:3rd_problem}
\end{align}
under the two constraints
\begin{align}
	\text{Transmit power: }&\sum_{m=0}^{M-1}|w_m|^2=1, \label{eq:transmit_power_constraint}\\  
	\text{Zero third-order distortion: }&\sum_{m=0}^{M-1} h_m w_m |w_m|^2=0.
\end{align}
Constraint (\ref{eq:transmit_power_constraint}) is not exactly a transmit power constraint as the nonlinear transmited power is disregarded. Given the saturation effect of the PA, it can be seen as an upper bound on the actual transmit power. This choice is made here for simplicity and for the sake of comparison with the \gls{mrt} precoder, which is found under a similar constraint. Moreover, this choice is a good approximation. Even if nonlinearities are considered in this work, we consider a regime where they may have a strong impact on the \gls{sdr} while having a limited impact on the total transmit power. As a simple example, consider an antenna with 4\% nonlinear distortion power with respect to the total transmit power. This might be negligible in terms of transmit power while the \gls{sdr} is limited to about $14$ dB only. It thus makes sense to neglect them in the transmit power computation but consider them in the \gls{sndr} expression.

The above problem is non-concave and not trivial to solve. However, it can be first reformulated in a simpler form using the change of variable $w_m=g_me^{-\jmath \angle h_m}$. Defining $r_m=|h_m|$, the reformulated problem becomes
\begin{align}
	\max_{g_0,...,g_{M-1}} \left|\sum_{m=0}^{M-1} r_m g_m\right|^2\ \text{s.t.}\ &\sum_{m=0}^{M-1} |g_m|^2=1,\label{eq:complex_formulation}\\ 
	&\sum_{m=0}^{M-1} r_m g_m |g_m|^2=0.\nonumber
\end{align}
For a given $g_m$, $w_m$ can be retrieved as $w_m=g_me^{-\jmath \angle h_m}$. From the above formulation, a conjecture can be made. In the following, to avoid equivalent symmetric solutions, we implicitly constrain the precoders to lead to a real and positive array gain $\sum_{m=0}^{M-1}r_m g_m$. Indeed, one can insert $w_m e^{\jmath \phi}$ with $\phi \in \mathbb{R}$ in (\ref{eq:3rd_problem}) and easily check that the array gain is equal to the one achieved by $w_m$ while (\ref{eq:transmit_power_constraint}) and (\ref{eq:zero_distortion_constraint}) still hold. From the formulation (\ref{eq:3rd_problem}), a conjecture can be made.
\begin{conjecture}\label{conjecture}
	An optimal $g_m$ for problem (\ref{eq:complex_formulation}) should be purely real up to a constant phasor.
\end{conjecture}
\begin{proof}
	This conjecture is not yet rigorously proven. Still, clarifying elements are provided in Appendix~\ref{appendix:conjecture}.
\end{proof}
Using this conjecture, the problem is converted to an all real problem
\begin{align}
	\max_{g_0,...,g_{M-1}} \left(\sum_{m=0}^{M-1} r_m g_m\right)^2 \text{s.t.} \sum_{m=0}^{M-1} g_m^2=1, \sum_{m=0}^{M-1} r_m g_m^3=0, \label{eq:all_real_3rd_problem}
\end{align}
In the following, we always assume that $r_m>0,\ \forall m$. Otherwise the zero gain antennas can be set inactive and discarded from the optimization. Moreover, we define $\alpha$ in a general sense as a normalization constant that ensures that the transmit power constraint is satisfied, \textit{i.e.}, for a given precoder, real $g_m$ or complex $w_m$, $\alpha$ is given by
\begin{align}
	\alpha=\frac{1}{\sqrt{\sum_{m=0}^{M-1}g_m^2}}=\frac{1}{\sqrt{\sum_{m=0}^{M-1}|w_m|^2}}. \label{eq:alpha_def}
\end{align}
The maxima of (\ref{eq:all_real_3rd_problem}) are given in the following theorem.
\begin{theorem}[Maxima]\label{theorem:maxima} Problem (\ref{eq:all_real_3rd_problem}) has $M$ potential maxima. Indexing them by $m'=0,...,M-1$, the $m'$-th maximum is feasible if there is a positive and real constant $\xi$ which is the solution of the equation
	\begin{align*}
		\sum_{m=0, m \neq m'}^{M-1} \frac{(-1+\sqrt{1+r_m^2\xi})^3}{r_m^2}=\frac{(1+\sqrt{1+r_{m'}^2\xi})^3}{r_{m'}^2}.
	\end{align*}
	The maximum is then obtained by using the precoder
	\begin{align*}
		g_{m,m'}&=\alpha \begin{cases}
			\frac{-1+\sqrt{1+r_m^2\xi}}{r_m}\ &\text{if}\ m \neq m'\\
			\frac{-1-\sqrt{1+r_m^2\xi}}{r_m}\ &\text{if}\ m = m'
		\end{cases}.
	\end{align*}
\end{theorem}
\begin{proof}
	See Appendix~\ref{appendix:proof_theorem}.
\end{proof}

\begin{remark}[Saturated antenna]
	Each maximum is obtained by using a single antenna with a negative gain and saturated in such a way that its nonlinear distortion cancels the aggregated nonlinear distortion of all other antennas at the user location. In other words, distortion is used and amplified to cancel distortion. There is a certain analogy possible with the tone reservation technique \cite{Krongold2004}. The main difference is that, here, the cancelling elements are in the space domain and not in the frequency domain. Hence, they do not lead to a loss of spectral utilization. They also rely on the saturation of the cancelling elements to be more energy-efficient.
\end{remark}

\begin{remark}[Global maximum]\label{remark:global_optimum}
	Finding the global optimum of (\ref{eq:all_real_3rd_problem}) requires finding the maximum among the $M$ maxima. This is studied via simulations in Section~\ref{section:Simulation_results} (Fig.~\ref{fig:comparison_maxima}). It is shown that the maxima are close the the global one as long as the saturated antenna has a gain $r_m$ close to the median of all $r_m$. For typical channel gains that we simulated, all maxima were found to be feasible. Non-feasible maxima were found in pathological cases where the channel gain of the saturated antenna $r_{m'}$ is very large, \textit{e.g.}, more than 20 dB above the average of the other ones.
\end{remark}

To avoid having to perform the line search procedure and having to search among the different maxima, we present in the following the so-called Z3RO precoder which has a closed-form expression. We also show that, while inducing an array gain penalty, using more than one antenna with negative gains can be useful in practice.

\subsection{Line-of-Sight Channel}

\begin{corollary}[Z3RO precoder in LOS]\label{corollary_LOS}
	For the pure \gls{los} channel given in (\ref{eq:channel_model_LOS}), all maxima of (\ref{eq:all_real_3rd_problem}) are global maxima, \textit{i.e.}, achieve the same array gain. The $m'$-th maxima is then given by
	\begin{align*}
		g_{m,m'}&=\alpha \begin{cases}
			1\ &\text{if}\ m \neq m'\\
			-(M-1)^{1/3}\ &\text{if}\ m = m'
		\end{cases},
	\end{align*}
	Moreover, critical points of (\ref{eq:all_real_3rd_problem}), are obtained by defining a set $\mathcal{M}$ of $M_s$ antennas, chosen arbitrarily among the $M$ antennas with $M/2>M_s>0$, and using the precoder
	\begin{align*}
		g_{m,\mathcal{M}}^{\mathrm{Z3RO}}&=\alpha \begin{cases}
			- \left(\frac{M-M_s}{M_s}\right)^{1/3}\ &\text{if}\ m \in \mathcal{M}\\
			1\ &\text{otherwise}
		\end{cases}.
	\end{align*}
	The SNR at the user is
	\begin{align}
		\text{SNR}_{\mathrm{LOS}}^{\mathrm{Z3RO}}&=\frac{\beta p}{\sigma_v^2} M \frac{\left(\zeta^{2/3}-(1-\zeta)^{2/3}\right)^{2}}{\zeta^{1/3}+(1-\zeta)^{1/3}}, \label{eq:SNR_Z3RO}
	\end{align}
	where $\zeta=M_s/M$. For a fixed value of $M_s$,
	\begin{align*}
		\lim_{M\rightarrow +\infty} \frac{\text{SNR}_{\mathrm{LOS}}^{\mathrm{Z3RO}}}{\text{SNR}^{\mathrm{MRT}}} = 1.
	\end{align*}
	The array gain penalty versus MRT vanishes.
\end{corollary}
\begin{proof}
	See Appendix~\ref{appendix:proof_corollary}.
\end{proof}

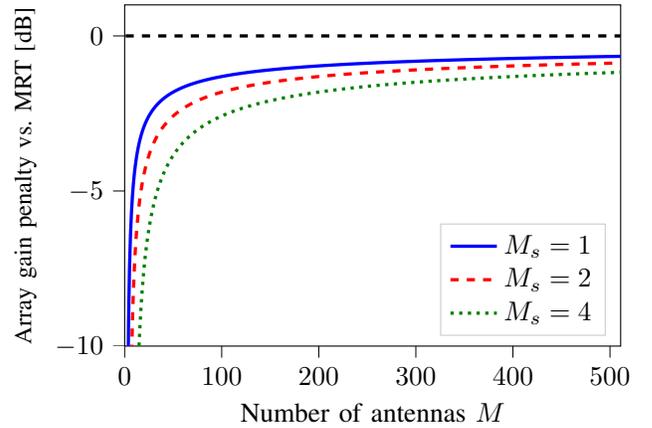
\begin{figure}[t!]
	\centering 
	\resizebox{\linewidthmod}{!}{%
		{\input{Fig/array_gain_penalty.tex}}
	} 
	\caption{As the number of antennas increases, the penalty in array gain of the Z3RO precoder vanishes, as compared to MRT. $M_s$ is the number of saturated antennas (with negative gains).}
	\label{fig:array_gain} 
\end{figure}
\begin{figure}[t!]
	\centering 
	\subfloat[Signal radiation pattern (dB), $M=32$.]{
		\resizebox{\linewidthmodd}{!}{%
			{\includegraphics[clip, trim=2.5cm 2cm 2.5cm 3cm, scale=1]{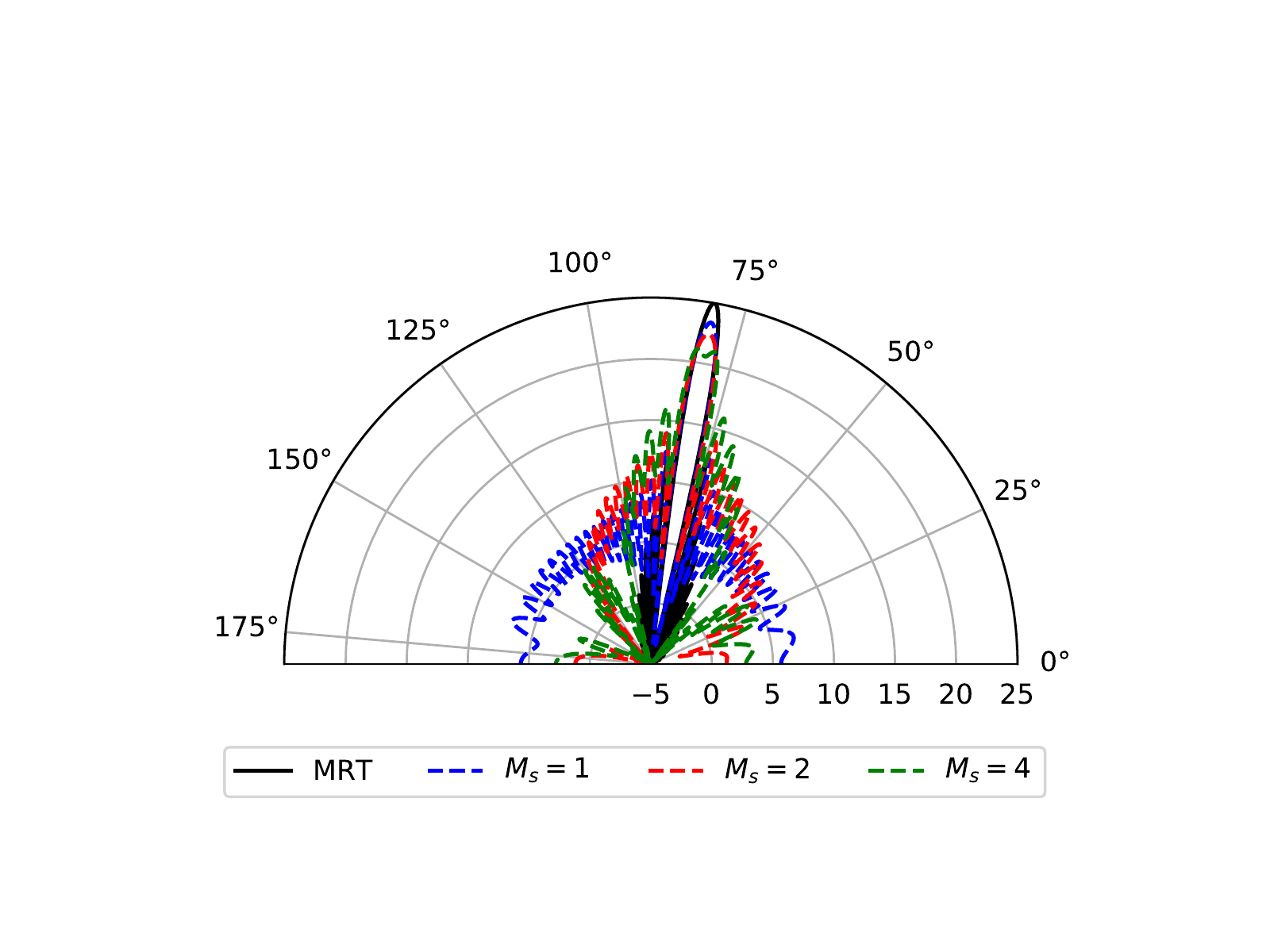}}
		}
	}
	
	\subfloat[Third-order distortion radiation pattern (dB), $M=32$.
	]{
		\resizebox{\linewidthmodd}{!}{%
			{\includegraphics[clip, trim=2.5cm 2cm 2.5cm 3cm, scale=1]{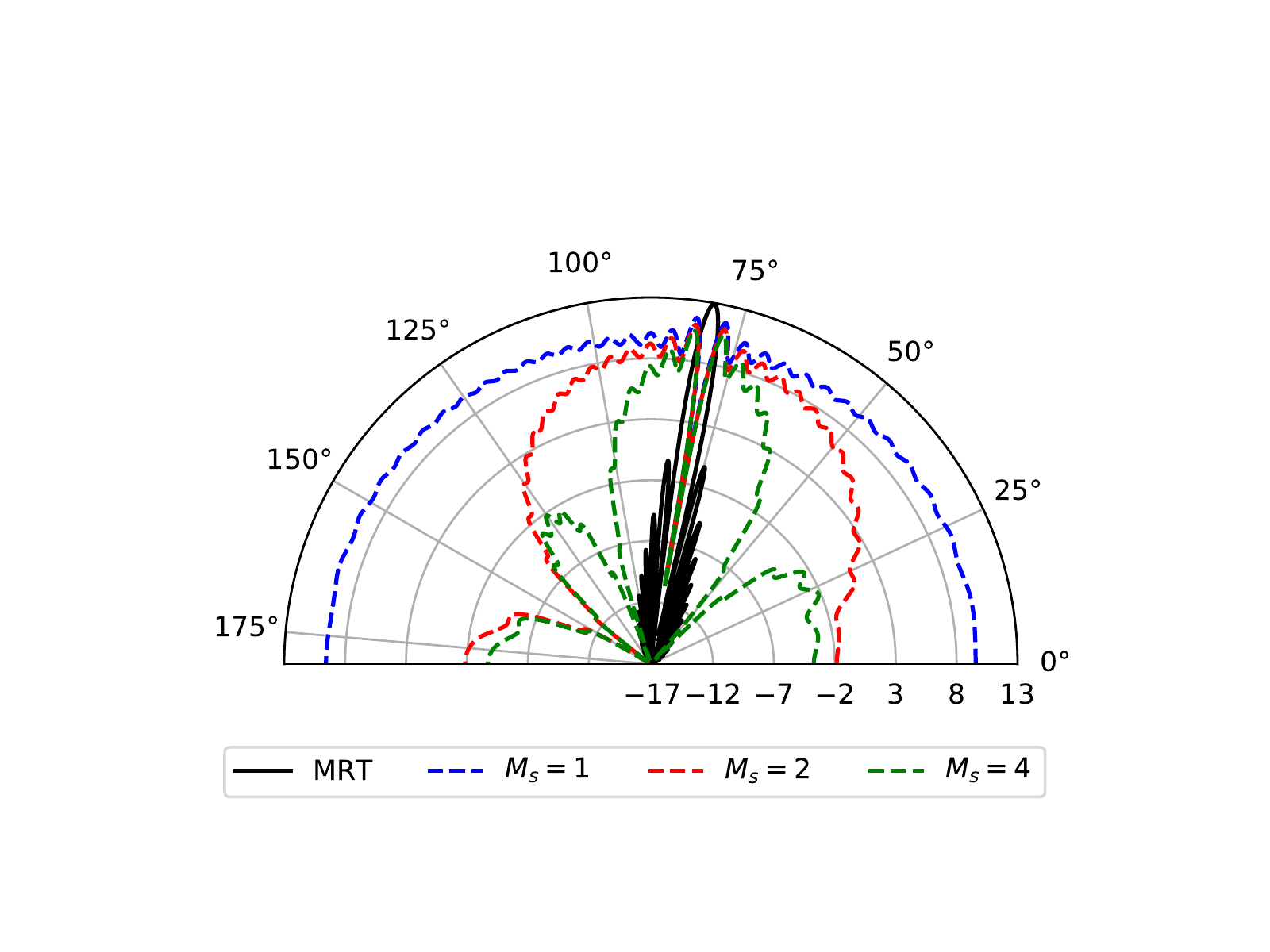}}
		}
	} 
	
	\caption{Radiation pattern for \gls{los} channel and half-wavelength \gls{ula}. As more antennas become saturated ($M_s\nearrow$), the array gain of the \gls{z3ro} precoder decreases. On the other hand, the total radiated power decreases and it becomes more spatially focused, which is beneficial regarding unintended directions.}
	\label{fig:radiation_patterns} 
\end{figure}

\begin{remark}[Saturated antennas]
	The critical points are obtained by using $M_s$ antennas with a negative gain and saturated in such a way that they compensate for the nonlinear distortion of all other antennas at the user location. Using $M_s=1$ gives the global optimum.
\end{remark}
\begin{remark}[Array gain]
	This precoder is proposed for large array systems, operating in the saturation regime, where SDR is limiting rather than SNR. As a result, the reduced array gain with respect to MRT becomes negligible. Moreover, as shown in Fig.~\ref{fig:array_gain}, this array gain penalty vanishes 	for large array systems, as $M$ grows large. For $M_s = 1$ and $M = 64$, the MRT precoder achieves an array gain of about 18 dB versus 16.5 dB for the proposed precoder, while the third distortion order is completely cancelled.
\end{remark}
\begin{remark}[Radiation pattern]
	Fig.~\ref{fig:directivity_patterns} shows directivity patterns of the signal and distortion for the MRT and the optimal precoder for $M_s=1$. On the other hand, Fig.~\ref{fig:radiation_patterns} shows their absolute radiation pattern for different values of $M_s$. In Fig.~\ref{fig:radiation_patterns}~(a), the array gain decreases when $M_s$ increases. This would imply only using the design $M_s = 1$. However, as shown in Fig.~\ref{fig:radiation_patterns}~(b), this design, as compared to MRT, leads to an increase of total distortion power, which is mainly radiated towards unintended locations. This leads to interference for potential observers, especially since PA nonlinearities induce out-of-band emissions. Hence, a user in an adjacent band could suffer from it. On the other hand, as $M_s$ increases, the distortion becomes focused ``approximately in the user direction, except for the exact user direction, where it is null by design". Moreover, the total radiated distortion power is reduced. This comes from the fact that distortion is beamformed and benefits from an array gain. As a result, choosing the value $M_s$ offers a trade-off between array gain and spatial focusing
	of the distortion.
\end{remark}
%

\subsection{General Channel}

Inspired by the structure of critical points in the LOS case, we propose a heuristic precoder for the general channel case.

\begin{heuristic}[Z3RO precoder]\label{heuristic_Z3RO}
	The following heuristic precoder achieves a good performance. Defining a set $\mathcal{M}$ of $M_s$ antennas, chosen arbitrarily among the $M$ antennas with $M/2>M_s>0$, choose
	\begin{align}
		g_{m,\mathcal{M}}^{\mathrm{Z3RO}}	&=\alpha r_m\begin{cases}
			-\left(\frac{\sum_{m'=M_s}^{M-1} r_{m'}^4}{\sum_{m''=0}^{M_s-1} r_{m''}^4}\right)^{1/3} &\text{if}\ m \in \mathcal{M}\\
			1 &\text{otherwise}
		\end{cases}. \label{eq:heuristic_precoder}
	\end{align}	
\end{heuristic}
It will be illustrated through numerical simulations that this precoder yields a negligible array gain penalty as compared to the optimal one. The following remark and proposition give insight on why the heuristic precoder (\ref{eq:heuristic_precoder}) performs well for LOS and general channels respectively.

\begin{proposition}[Z3RO performance in LOS]
	For the pure \gls{los} channel given in (\ref{eq:channel_model_LOS}), the Z3RO precoder (\ref{eq:heuristic_precoder}) reduces to the one of Corollary~\ref{corollary_LOS}. This explains why the same precoder notation $g_{m,\mathcal{M}}^{\mathrm{Z3RO}}$ has been used. This also implies that, for $M_s=1$, it is optimal.
\end{proposition}
\begin{proof}
	The proof results from a straightforward particularization to the LOS case $r_m=\sqrt{\beta}$.
\end{proof}

\begin{proposition}[Z3RO performance in large antenna systems] \label{proposition:Z3RO_large_antenna}
	If the channel gains $r_m$ are independent and identically distributed (i.i.d.) with variance $\beta$, as $M$ and $M_s$ grow large with a fixed ratio $\zeta=M_s/M$, 
	\begin{align*}
		\frac{\text{SNR}^{\mathrm{Z3RO}}}{\text{SNR}_{\mathrm{LOS}}^{\mathrm{Z3RO}}} \rightarrow 1,
	\end{align*}
	where $\text{SNR}^{\mathrm{Z3RO}}$ is the SNR achieved by the heuristic precoder (\ref{eq:heuristic_precoder}) and $\text{SNR}_{\mathrm{LOS}}^{\mathrm{Z3RO}}$ is the one in the deterministic LOS case (\ref{eq:SNR_Z3RO}).
\end{proposition}
\begin{proof}
	See Appendix~\ref{appendix:proof_proposition}.
\end{proof}

The previous propositions imply that the performance of the heuristic Z3RO precoder for a general channel and a large antenna system are close to the one of the optimal precoder in LOS channel.

%
%

%% file: Fig/array_gain_penalty.tex
\begin{tikzpicture}

\begin{axis}[
width=8cm,
height=6cm,
legend cell align={left},
legend style={
  fill opacity=0.8,
  draw opacity=1,
  text opacity=1,
  at={(0.97,0.03)},
  anchor=south east,
  draw=white!80!black
},
tick align=outside,
tick pos=left,
x grid style={white!69.0196078431373!black},
xlabel={Number of antennas \(\displaystyle M\)},
xmin=0, xmax=511,
xtick style={color=black},
y grid style={white!69.0196078431373!black},
ylabel style={align=center},
ylabel={\small Array gain penalty vs. MRT [dB]},
ymin=-10, ymax=1,
ytick style={color=black}
]
\addplot [very thick, black, dashed, forget plot]
table {%
1 0
2 0
3 0
4 0
5 0
6 0
7 0
8 0
9 0
10 0
11 0
12 0
13 0
14 0
15 0
16 0
17 0
18 0
19 0
20 0
21 0
22 0
23 0
24 0
25 0
26 0
27 0
28 0
29 0
30 0
31 0
32 0
33 0
34 0
35 0
36 0
37 0
38 0
39 0
40 0
41 0
42 0
43 0
44 0
45 0
46 0
47 0
48 0
49 0
50 0
51 0
52 0
53 0
54 0
55 0
56 0
57 0
58 0
59 0
60 0
61 0
62 0
63 0
64 0
65 0
66 0
67 0
68 0
69 0
70 0
71 0
72 0
73 0
74 0
75 0
76 0
77 0
78 0
79 0
80 0
81 0
82 0
83 0
84 0
85 0
86 0
87 0
88 0
89 0
90 0
91 0
92 0
93 0
94 0
95 0
96 0
97 0
98 0
99 0
100 0
101 0
102 0
103 0
104 0
105 0
106 0
107 0
108 0
109 0
110 0
111 0
112 0
113 0
114 0
115 0
116 0
117 0
118 0
119 0
120 0
121 0
122 0
123 0
124 0
125 0
126 0
127 0
128 0
129 0
130 0
131 0
132 0
133 0
134 0
135 0
136 0
137 0
138 0
139 0
140 0
141 0
142 0
143 0
144 0
145 0
146 0
147 0
148 0
149 0
150 0
151 0
152 0
153 0
154 0
155 0
156 0
157 0
158 0
159 0
160 0
161 0
162 0
163 0
164 0
165 0
166 0
167 0
168 0
169 0
170 0
171 0
172 0
173 0
174 0
175 0
176 0
177 0
178 0
179 0
180 0
181 0
182 0
183 0
184 0
185 0
186 0
187 0
188 0
189 0
190 0
191 0
192 0
193 0
194 0
195 0
196 0
197 0
198 0
199 0
200 0
201 0
202 0
203 0
204 0
205 0
206 0
207 0
208 0
209 0
210 0
211 0
212 0
213 0
214 0
215 0
216 0
217 0
218 0
219 0
220 0
221 0
222 0
223 0
224 0
225 0
226 0
227 0
228 0
229 0
230 0
231 0
232 0
233 0
234 0
235 0
236 0
237 0
238 0
239 0
240 0
241 0
242 0
243 0
244 0
245 0
246 0
247 0
248 0
249 0
250 0
251 0
252 0
253 0
254 0
255 0
256 0
257 0
258 0
259 0
260 0
261 0
262 0
263 0
264 0
265 0
266 0
267 0
268 0
269 0
270 0
271 0
272 0
273 0
274 0
275 0
276 0
277 0
278 0
279 0
280 0
281 0
282 0
283 0
284 0
285 0
286 0
287 0
288 0
289 0
290 0
291 0
292 0
293 0
294 0
295 0
296 0
297 0
298 0
299 0
300 0
301 0
302 0
303 0
304 0
305 0
306 0
307 0
308 0
309 0
310 0
311 0
312 0
313 0
314 0
315 0
316 0
317 0
318 0
319 0
320 0
321 0
322 0
323 0
324 0
325 0
326 0
327 0
328 0
329 0
330 0
331 0
332 0
333 0
334 0
335 0
336 0
337 0
338 0
339 0
340 0
341 0
342 0
343 0
344 0
345 0
346 0
347 0
348 0
349 0
350 0
351 0
352 0
353 0
354 0
355 0
356 0
357 0
358 0
359 0
360 0
361 0
362 0
363 0
364 0
365 0
366 0
367 0
368 0
369 0
370 0
371 0
372 0
373 0
374 0
375 0
376 0
377 0
378 0
379 0
380 0
381 0
382 0
383 0
384 0
385 0
386 0
387 0
388 0
389 0
390 0
391 0
392 0
393 0
394 0
395 0
396 0
397 0
398 0
399 0
400 0
401 0
402 0
403 0
404 0
405 0
406 0
407 0
408 0
409 0
410 0
411 0
412 0
413 0
414 0
415 0
416 0
417 0
418 0
419 0
420 0
421 0
422 0
423 0
424 0
425 0
426 0
427 0
428 0
429 0
430 0
431 0
432 0
433 0
434 0
435 0
436 0
437 0
438 0
439 0
440 0
441 0
442 0
443 0
444 0
445 0
446 0
447 0
448 0
449 0
450 0
451 0
452 0
453 0
454 0
455 0
456 0
457 0
458 0
459 0
460 0
461 0
462 0
463 0
464 0
465 0
466 0
467 0
468 0
469 0
470 0
471 0
472 0
473 0
474 0
475 0
476 0
477 0
478 0
479 0
480 0
481 0
482 0
483 0
484 0
485 0
486 0
487 0
488 0
489 0
490 0
491 0
492 0
493 0
494 0
495 0
496 0
497 0
498 0
499 0
500 0
501 0
502 0
503 0
504 0
505 0
506 0
507 0
508 0
509 0
510 0
511 0
};
\addplot [very thick, blue]
table {%
1 -inf
2 -inf
3 -12.9334508218125
4 -9.22935110225897
5 -7.48236814032019
6 -6.42698543269992
7 -5.70720290536271
8 -5.1788367948702
9 -4.77121254719662
10 -4.44522541248991
11 -4.17732407920397
12 -3.95240187712169
13 -3.76028383200591
14 -3.59384384166512
15 -3.44792903261676
16 -3.31871268647317
17 -3.20328799235772
18 -3.09940365467987
19 -3.00528645292255
20 -2.91951895009947
21 -2.84095322967714
22 -2.76864878815991
23 -2.70182699853211
24 -2.63983717570442
25 -2.58213091498725
26 -2.52824242789335
27 -2.47777329103335
28 -2.43038048686295
29 -2.38576693068035
30 -2.34367389701513
31 -2.30387491244054
32 -2.26617079161743
33 -2.23038557269412
34 -2.19636316617718
35 -2.16396457425776
36 -2.13306556959523
37 -2.10355474670289
38 -2.07533187744989
39 -2.04830651628762
40 -2.02239681170949
41 -1.99752848894596
42 -1.97363397556092
43 -1.95065164687981
44 -1.92852517236396
45 -1.90720294739149
46 -1.88663759759624
47 -1.86678554509084
48 -1.84760662766883
49 -1.82906376352383
50 -1.81112265520853
51 -1.79375152753246
52 -1.7769208949053
53 -1.76060335430389
54 -1.74477340060117
55 -1.72940726146416
56 -1.71448274942251
57 -1.69997912904089
58 -1.68587699741048
59 -1.67215817641264
60 -1.65880561541184
61 -1.64580330320756
62 -1.63313618822438
63 -1.62079010604611
64 -1.60875171351049
65 -1.59700842867512
66 -1.58554837604783
67 -1.57436033654539
68 -1.56343370170679
69 -1.55275843174085
70 -1.54232501703511
71 -1.5321244427943
72 -1.52214815651242
73 -1.51238803801477
74 -1.5028363718335
75 -1.49348582170543
76 -1.48432940700259
77 -1.47536048092474
78 -1.46657271030094
79 -1.4579600568616
80 -1.44951675985673
81 -1.44123731990746
82 -1.433116483989
83 -1.42514923145272
84 -1.41733076100351
85 -1.40965647855647
86 -1.40212198590352
87 -1.39472307012721
88 -1.38745569370406
89 -1.38031598524524
90 -1.37330023082665
91 -1.3664048658646
92 -1.35962646749716
93 -1.35296174743425
94 -1.34640754524307
95 -1.33996082203778
96 -1.33361865454495
97 -1.32737822951895
98 -1.32123683848279
99 -1.31519187277264
100 -1.30924081886527
101 -1.30338125396975
102 -1.29761084186577
103 -1.29192732897263
104 -1.28632854063381
105 -1.28081237760346
106 -1.27537681272183
107 -1.27001988776791
108 -1.26473971047808
109 -1.2595344517207
110 -1.254402342817
111 -1.24934167299942
112 -1.24435078699921
113 -1.23942808275549
114 -1.23457200923879
115 -1.22978106438222
116 -1.22505379311408
117 -1.22038878548619
118 -1.21578467489234
119 -1.21124013637191
120 -1.20675388499379
121 -1.20232467431627
122 -1.19795129491859
123 -1.19363257300036
124 -1.18936736904508
125 -1.18515457654441
126 -1.18099312077994
127 -1.17688195765932
128 -1.17282007260409
129 -1.16880647948638
130 -1.16484021961202
131 -1.16092036074766
132 -1.15704599618972
133 -1.15321624387301
134 -1.14943024551706
135 -1.1456871658083
136 -1.14198619161628
137 -1.13832653124239
138 -1.1347074136993
139 -1.1311280880199
140 -1.12758782259404
141 -1.12408590453205
142 -1.12062163905351
143 -1.11719434890036
144 -1.11380337377289
145 -1.11044806978789
146 -1.10712780895774
147 -1.10384197868955
148 -1.10058998130344
149 -1.09737123356908
150 -1.09418516625978
151 -1.09103122372316
152 -1.08790886346786
153 -1.08481755576553
154 -1.08175678326733
155 -1.07872604063457
156 -1.07572483418251
157 -1.07275268153722
158 -1.06980911130447
159 -1.06689366275056
160 -1.06400588549433
161 -1.06114533920994
162 -1.05831159334011
163 -1.05550422681919
164 -1.0527228278058
165 -1.04996699342462
166 -1.04723632951693
167 -1.04453045039961
168 -1.0418489786322
169 -1.03919154479174
170 -1.0365577872551
171 -1.03394735198845
172 -1.03135989234367
173 -1.02879506886138
174 -1.02625254908032
175 -1.0237320073529
176 -1.02123312466667
177 -1.01875558847145
178 -1.01629909251187
179 -1.01386333666536
180 -1.01144802678498
181 -1.0090528745473
182 -1.00667759730495
183 -1.00432191794368
184 -1.00198556474379
185 -0.99966827124584
186 -0.997369776120338
187 -0.995089823041408
188 -0.99282816056422
189 -0.990584542006061
190 -0.988358725330976
191 -0.986150473037696
192 -0.983959552050999
193 -0.9817857336161
194 -0.979628793196207
195 -0.977488510372961
196 -0.975364668749754
197 -0.973257055857824
198 -0.971165463064988
199 -0.969089685486989
200 -0.967029521901283
201 -0.964984774663285
202 -0.962955249624917
203 -0.96094075605542
204 -0.958941106564349
205 -0.95695611702667
206 -0.954985606509919
207 -0.953029397203293
208 -0.95108731434871
209 -0.9491591861737
210 -0.947244843826075
211 -0.945344121310335
212 -0.943456855425816
213 -0.941582885706391
214 -0.939722054361783
215 -0.937874206220447
216 -0.936039188673866
217 -0.934216851622349
218 -0.932407047422194
219 -0.930609630834193
220 -0.9288244589735
221 -0.927051391260708
222 -0.925290289374212
223 -0.923541017203754
224 -0.921803440805102
225 -0.920077428355918
226 -0.918362850112667
227 -0.916659578368624
228 -0.914967487412882
229 -0.913286453490409
230 -0.911616354763041
231 -0.909957071271425
232 -0.908308484897891
233 -0.906670479330224
234 -0.905042940026294
235 -0.903425754179529
236 -0.901818810685199
237 -0.900222000107519
238 -0.898635214647502
239 -0.897058348111611
240 -0.895491295881051
241 -0.893933954881903
242 -0.892386223555811
243 -0.89084800183147
244 -0.889319191096684
245 -0.887799694171074
246 -0.886289415279479
247 -0.884788260025834
248 -0.88329613536776
249 -0.881812949591649
250 -0.880338612288344
251 -0.878873034329334
252 -0.877416127843506
253 -0.875967806194384
254 -0.87452798395789
255 -0.873096576900593
256 -0.871673501958411
257 -0.8702586772158
258 -0.868852021885384
259 -0.867453456288029
260 -0.866062901833349
261 -0.864680281000603
262 -0.863305517320052
263 -0.861938535354653
264 -0.860579260682187
265 -0.859227619877743
266 -0.857883540496576
267 -0.856546951057311
268 -0.855217781025516
269 -0.85389596079761
270 -0.852581421685106
271 -0.851274095899155
272 -0.849973916535467
273 -0.848680817559454
274 -0.847394733791762
275 -0.846115600894049
276 -0.844843355355048
277 -0.843577934476933
278 -0.84231927636196
279 -0.841067319899357
280 -0.839822004752474
281 -0.83858327134622
282 -0.837351060854706
283 -0.836125315189168
284 -0.83490597698612
285 -0.833692989595701
286 -0.832486297070317
287 -0.831285844153447
288 -0.830091576268684
289 -0.828903439509011
290 -0.827721380626242
291 -0.826545347020698
292 -0.82537528673106
293 -0.824211148424437
294 -0.823052881386599
295 -0.821900435512413
296 -0.820753761296453
297 -0.819612809823764
298 -0.81847753276085
299 -0.817347882346791
300 -0.816223811384534
301 -0.815105273232331
302 -0.81399222179538
303 -0.812884611517565
304 -0.811782397373386
305 -0.810685534860029
306 -0.809593979989554
307 -0.808507689281265
308 -0.807426619754197
309 -0.80635072891974
310 -0.805279974774388
311 -0.804214315792637
312 -0.803153710919995
313 -0.802098119566127
314 -0.801047501598107
315 -0.800001817333809
316 -0.798961027535394
317 -0.797925093402924
318 -0.796893976568094
319 -0.795867639088062
320 -0.794846043439374
321 -0.793829152512035
322 -0.792816929603638
323 -0.791809338413627
324 -0.790806343037639
325 -0.78980790796195
326 -0.788813998058019
327 -0.787824578577125
328 -0.786839615145076
329 -0.785859073757039
330 -0.784882920772452
331 -0.78391112290997
332 -0.78294364724259
333 -0.78198046119276
334 -0.781021532527641
335 -0.780066829354421
336 -0.779116320115683
337 -0.778169973584906
338 -0.777227758861985
339 -0.776289645368851
340 -0.775355602845169
341 -0.774425601344081
342 -0.773499611228048
343 -0.772577603164742
344 -0.771659548122994
345 -0.770745417368831
346 -0.769835182461586
347 -0.768928815250005
348 -0.768026287868513
349 -0.767127572733449
350 -0.766232642539431
351 -0.765341470255722
352 -0.764454029122705
353 -0.76357029264837
354 -0.762690234604882
355 -0.761813829025197
356 -0.760941050199725
357 -0.760071872673057
358 -0.759206271240725
359 -0.758344220946032
360 -0.757485697076922
361 -0.756630675162886
362 -0.755779130971936
363 -0.754931040507626
364 -0.754086380006092
365 -0.753245125933161
366 -0.752407254981503
367 -0.751572744067817
368 -0.750741570330059
369 -0.749913711124728
370 -0.749089144024159
371 -0.748267846813899
372 -0.747449797490095
373 -0.746634974256931
374 -0.745823355524086
375 -0.745014919904261
376 -0.744209646210706
377 -0.743407513454832
378 -0.74260850084379
379 -0.741812587778158
380 -0.741019753849606
381 -0.740229978838632
382 -0.739443242712295
383 -0.738659525622037
384 -0.737878807901457
385 -0.737101070064202
386 -0.736326292801816
387 -0.735554456981681
388 -0.734785543644928
389 -0.734019534004445
390 -0.733256409442833
391 -0.732496151510466
392 -0.731738741923539
393 -0.730984162562152
394 -0.730232395468406
395 -0.72948342284458
396 -0.728737227051238
397 -0.727993790605467
398 -0.727253096179075
399 -0.726515126596808
400 -0.725779864834663
401 -0.725047294018112
402 -0.724317397420483
403 -0.723590158461236
404 -0.72286556070436
405 -0.722143587856716
406 -0.721424223766483
407 -0.720707452421563
408 -0.719993257947999
409 -0.719281624608488
410 -0.718572536800835
411 -0.717865979056488
412 -0.717161936039036
413 -0.716460392542779
414 -0.715761333491288
415 -0.715064743935985
416 -0.714370609054777
417 -0.713678914150631
418 -0.712989644650259
419 -0.712302786102753
420 -0.711618324178283
421 -0.710936244666789
422 -0.710256533476654
423 -0.709579176633504
424 -0.708904160278889
425 -0.708231470669088
426 -0.707561094173868
427 -0.706893017275272
428 -0.706227226566459
429 -0.705563708750492
430 -0.704902450639214
431 -0.704243439152063
432 -0.703586661314993
433 -0.702932104259313
434 -0.702279755220601
435 -0.701629601537645
436 -0.700981630651311
437 -0.700335830103549
438 -0.699692187536294
439 -0.699050690690462
440 -0.69841132740494
441 -0.697774085615532
442 -0.697138953354008
443 -0.696505918747123
444 -0.695874970015618
445 -0.695246095473268
446 -0.69461928352598
447 -0.693994522670793
448 -0.693371801495
449 -0.69275110867523
450 -0.69213243297654
451 -0.691515763251521
452 -0.690901088439449
453 -0.690288397565372
454 -0.68967767973931
455 -0.689068924155352
456 -0.688462120090862
457 -0.687857256905646
458 -0.687254324041103
459 -0.686653311019466
460 -0.686054207442993
461 -0.68545700299315
462 -0.684861687429873
463 -0.684268250590789
464 -0.683676682390434
465 -0.683086972819541
466 -0.68249911194427
467 -0.681913089905503
468 -0.681328896918095
469 -0.680746523270166
470 -0.680165959322406
471 -0.679587195507373
472 -0.679010222328788
473 -0.678435030360889
474 -0.677861610247715
475 -0.677289952702465
476 -0.676720048506846
477 -0.676151888510404
478 -0.67558546362988
479 -0.675020764848605
480 -0.674457783215835
481 -0.673896509846159
482 -0.673336935918859
483 -0.67277905267733
484 -0.672222851428458
485 -0.671668323542055
486 -0.671115460450234
487 -0.670564253646861
488 -0.670014694686977
489 -0.669466775186216
490 -0.668920486820268
491 -0.668375821324306
492 -0.667832770492446
493 -0.667291326177217
494 -0.666751480289002
495 -0.666213224795537
496 -0.665676551721358
497 -0.665141453147301
498 -0.664607921210008
499 -0.664075948101371
500 -0.663545526068084
501 -0.663016647411124
502 -0.662489304485254
503 -0.66196348969855
504 -0.661439195511949
505 -0.660916414438708
506 -0.660395139043999
507 -0.659875361944417
508 -0.659357075807542
509 -0.658840273351441
510 -0.658324947344277
511 -0.657811090603845
};
\addlegendentry{$M_{s}=1$}
\addplot [very thick, red, dashed]
table {%
1 -60
2 -63.0102999566398
3 -64.7712125471966
4 -inf
5 -17.4554802119121
6 -12.9334508218125
7 -10.6582567012089
8 -9.22935110225897
9 -8.22952822294483
10 -7.48236814032019
11 -6.89845478815926
12 -6.42698543269992
13 -6.03670266027816
14 -5.70720290536272
15 -5.42454261720888
16 -5.1788367948702
17 -4.96286361603497
18 -4.77121254719662
19 -4.59974254643724
20 -4.44522541248991
21 -4.30510399997189
22 -4.17732407920397
23 -4.06021477724513
24 -3.95240187712169
25 -3.85274383709955
26 -3.76028383200591
27 -3.67421329432605
28 -3.59384384166512
29 -3.51858540902188
30 -3.44792903261675
31 -3.38143316309796
32 -3.31871268647317
33 -3.25943004370742
34 -3.20328799235772
35 -3.15002366427581
36 -3.09940365467987
37 -3.0512199382235
38 -3.00528645292255
39 -2.96143622703196
40 -2.91951895009947
41 -2.87939890953968
42 -2.84095322967714
43 -2.8040703624001
44 -2.76864878815991
45 -2.73459589364763
46 -2.70182699853211
47 -2.67026450849464
48 -2.63983717570442
49 -2.61047945104582
50 -2.58213091498725
51 -2.55473577609049
52 -2.52824242789335
53 -2.5026030563292
54 -2.47777329103335
55 -2.45371189487308
56 -2.43038048686294
57 -2.40774329431808
58 -2.38576693068034
59 -2.36442019594292
60 -2.34367389701513
61 -2.32350068572281
62 -2.30387491244053
63 -2.28477249360978
64 -2.26617079161743
65 -2.24804850569899
66 -2.23038557269412
67 -2.2131630766234
68 -2.19636316617718
69 -2.17996897931399
70 -2.16396457425776
71 -2.14833486626402
72 -2.13306556959523
73 -2.11814314420726
74 -2.1035547467029
75 -2.08928818515575
76 -2.07533187744989
77 -2.06167481281724
78 -2.04830651628762
79 -2.03521701579503
80 -2.02239681170949
81 -2.00983684858642
82 -1.99752848894596
83 -1.98546348891249
84 -1.97363397556092
85 -1.96203242583044
86 -1.95065164687982
87 -1.93948475776909
88 -1.92852517236396
89 -1.91776658336725
90 -1.9072029473915
91 -1.89682847099309
92 -1.88663759759624
93 -1.87662499524037
94 -1.86678554509084
95 -1.85711433065724
96 -1.84760662766883
97 -1.8382578945601
98 -1.82906376352383
99 -1.82002003209202
100 -1.81112265520853
101 -1.80236773775971
102 -1.79375152753245
103 -1.78527040857085
104 -1.7769208949053
105 -1.76869962462966
106 -1.76060335430389
107 -1.75262895366132
108 -1.74477340060116
109 -1.73703377644838
110 -1.72940726146416
111 -1.72189113059168
112 -1.71448274942251
113 -1.70717957037055
114 -1.69997912904089
115 -1.692879040782
116 -1.68587699741047
117 -1.6789707640982
118 -1.67215817641264
119 -1.6654371375012
120 -1.65880561541183
121 -1.65226164054176
122 -1.64580330320757
123 -1.63942875132956
124 -1.63313618822438
125 -1.62692387049978
126 -1.62079010604611
127 -1.6147332521193
128 -1.60875171351049
129 -1.60284394079767
130 -1.59700842867512
131 -1.59124371435649
132 -1.58554837604783
133 -1.57992103148688
134 -1.57436033654539
135 -1.56886498389114
136 -1.56343370170679
137 -1.55806525246267
138 -1.55275843174086
139 -1.54751206710809
140 -1.54232501703511
141 -1.5371961698601
142 -1.53212444279429
143 -1.5271087809676
144 -1.52214815651242
145 -1.51724156768389
146 -1.51238803801477
147 -1.50758661550344
148 -1.5028363718335
149 -1.49813640162349
150 -1.49348582170543
151 -1.48888377043075
152 -1.48432940700259
153 -1.47982191083308
154 -1.47536048092474
155 -1.47094433527471
156 -1.46657271030094
157 -1.46224486028945
158 -1.4579600568616
159 -1.45371758846065
160 -1.44951675985673
161 -1.44535689166946
162 -1.44123731990746
163 -1.43715739552416
164 -1.433116483989
165 -1.42911396487371
166 -1.42514923145272
167 -1.42122169031741
168 -1.41733076100351
169 -1.41347587563103
170 -1.40965647855647
171 -1.40587202603661
172 -1.40212198590352
173 -1.3984058372503
174 -1.39472307012721
175 -1.39107318524771
176 -1.38745569370406
177 -1.38387011669218
178 -1.38031598524524
179 -1.3767928399759
180 -1.37330023082665
181 -1.36983771682804
182 -1.3664048658646
183 -1.36300125444798
184 -1.35962646749716
185 -1.35628009812549
186 -1.35296174743425
187 -1.34967102431249
188 -1.34640754524307
189 -1.34317093411435
190 -1.33996082203777
191 -1.33677684717077
192 -1.33361865454495
193 -1.33048589589942
194 -1.32737822951895
195 -1.32429532007691
196 -1.32123683848279
197 -1.31820246173412
198 -1.31519187277264
199 -1.31220476034464
200 -1.30924081886527
201 -1.30629974828665
202 -1.30338125396975
203 -1.30048504655993
204 -1.29761084186578
205 -1.29475836074154
206 -1.29192732897263
207 -1.28911747716438
208 -1.28632854063381
209 -1.28356025930436
210 -1.28081237760346
211 -1.2780846443629
212 -1.27537681272183
213 -1.27268864003241
214 -1.27001988776791
215 -1.26737032143328
216 -1.26473971047808
217 -1.26212782821169
218 -1.2595344517207
219 -1.25695936178853
220 -1.254402342817
221 -1.25186318275008
222 -1.24934167299943
223 -1.24683760837197
224 -1.24435078699921
225 -1.24188101026843
226 -1.23942808275549
227 -1.23699181215947
228 -1.23457200923879
229 -1.23216848774904
230 -1.22978106438222
231 -1.22740955870764
232 -1.22505379311408
233 -1.22271359275354
234 -1.22038878548619
235 -1.21807920182679
236 -1.21578467489234
237 -1.21350504035092
238 -1.2112401363719
239 -1.20898980357721
240 -1.20675388499379
241 -1.20453222600719
242 -1.20232467431627
243 -1.20013107988886
244 -1.19795129491859
245 -1.19578517378262
246 -1.19363257300036
247 -1.1914933511932
248 -1.18936736904508
249 -1.187254489264
250 -1.18515457654442
251 -1.18306749753045
252 -1.18099312077995
253 -1.17893131672935
254 -1.17688195765932
255 -1.1748449176612
256 -1.17282007260409
257 -1.1708073001028
258 -1.16880647948638
259 -1.16681749176743
260 -1.16484021961201
261 -1.16287454731028
262 -1.16092036074766
263 -1.15897754737673
264 -1.15704599618972
265 -1.15512559769144
266 -1.15321624387301
267 -1.15131782818598
268 -1.14943024551706
269 -1.14755339216339
270 -1.1456871658083
271 -1.14383146549758
272 -1.14198619161628
273 -1.14015124586596
274 -1.13832653124239
275 -1.13651195201376
276 -1.1347074136993
277 -1.13291282304837
278 -1.1311280880199
279 -1.12935311776236
280 -1.12758782259404
281 -1.12583211398378
282 -1.12408590453204
283 -1.12234910795244
284 -1.12062163905351
285 -1.11890341372101
286 -1.11719434890036
287 -1.11549436257967
288 -1.11380337377289
289 -1.11212130250339
290 -1.11044806978789
291 -1.10878359762059
292 -1.10712780895774
293 -1.1054806277024
294 -1.10384197868955
295 -1.10221178767148
296 -1.10058998130344
297 -1.09897648712956
298 -1.09737123356908
299 -1.09577414990281
300 -1.09418516625979
301 -1.09260421360432
302 -1.09103122372316
303 -1.08946612921294
304 -1.08790886346786
305 -1.08635936066763
306 -1.08481755576553
307 -1.08328338447683
308 -1.08175678326733
309 -1.08023768934212
310 -1.07872604063457
311 -1.07722177579548
312 -1.07572483418251
313 -1.0742351558497
314 -1.07275268153722
315 -1.07127735266132
316 -1.06980911130447
317 -1.06834790020559
318 -1.06689366275056
319 -1.06544634296286
320 -1.06400588549433
321 -1.06257223561613
322 -1.06114533920994
323 -1.05972514275909
324 -1.05831159334011
325 -1.05690463861424
326 -1.05550422681919
327 -1.05411030676094
328 -1.0527228278058
329 -1.05134173987253
330 -1.04996699342462
331 -1.04859853946266
332 -1.04723632951693
333 -1.04588031564002
334 -1.04453045039961
335 -1.04318668687143
336 -1.0418489786322
337 -1.04051727975286
338 -1.03919154479174
339 -1.03787172878801
340 -1.0365577872551
341 -1.03524967617429
342 -1.03394735198845
343 -1.03265077159576
344 -1.03135989234367
345 -1.03007467202287
346 -1.02879506886138
347 -1.02752104151874
348 -1.02625254908032
349 -1.02498955105164
350 -1.0237320073529
351 -1.0224798783135
352 -1.02123312466667
353 -1.01999170754426
354 -1.01875558847145
355 -1.01752472936173
356 -1.01629909251187
357 -1.01507864059692
358 -1.01386333666536
359 -1.01265314413436
360 -1.01144802678498
361 -1.0102479487576
362 -1.0090528745473
363 -1.00786276899939
364 -1.00667759730495
365 -1.0054973249965
366 -1.00432191794368
367 -1.00315134234901
368 -1.00198556474379
369 -1.0008245519839
370 -0.99966827124584
371 -0.998516690022717
372 -0.997369776120338
373 -0.996227497653341
374 -0.995089823041408
375 -0.99395672100551
376 -0.99282816056422
377 -0.991704111030089
378 -0.990584542006063
379 -0.98946942338196
380 -0.988358725330974
381 -0.987252418306292
382 -0.986150473037696
383 -0.98505286052824
384 -0.983959552050999
385 -0.982870519145797
386 -0.9817857336161
387 -0.980705167525805
388 -0.979628793196207
389 -0.978556583202937
390 -0.977488510372961
391 -0.976424547781622
392 -0.975364668749755
393 -0.974308846840757
394 -0.973257055857824
395 -0.972209269841101
396 -0.971165463064986
397 -0.970125610035371
398 -0.969089685486991
399 -0.968057664380782
400 -0.967029521901282
401 -0.966005233454086
402 -0.964984774663285
403 -0.963968121368999
404 -0.962955249624917
405 -0.961946135695858
406 -0.96094075605542
407 -0.959939087383564
408 -0.958941106564349
409 -0.957946790683599
410 -0.95695611702667
411 -0.9559690630762
412 -0.954985606509919
413 -0.954005725198469
414 -0.953029397203293
415 -0.952056600774475
416 -0.95108731434871
417 -0.950121516547213
418 -0.9491591861737
419 -0.948200302212404
420 -0.947244843826075
421 -0.946292790354046
422 -0.945344121310335
423 -0.944398816381706
424 -0.943456855425816
425 -0.9425182184694
426 -0.94158288570639
427 -0.940650837496174
428 -0.939722054361783
429 -0.938796516988167
430 -0.937874206220447
431 -0.93695510306222
432 -0.936039188673866
433 -0.93512644437091
434 -0.934216851622349
435 -0.933310392049063
436 -0.932407047422194
437 -0.931506799661576
438 -0.930609630834193
439 -0.929715523152612
440 -0.9288244589735
441 -0.927936420796084
442 -0.927051391260708
443 -0.926169353147358
444 -0.925290289374212
445 -0.92441418299624
446 -0.923541017203754
447 -0.922670775321071
448 -0.921803440805102
449 -0.920938997244022
450 -0.920077428355918
451 -0.919218717987471
452 -0.918362850112667
453 -0.917509808831477
454 -0.916659578368624
455 -0.915812143072283
456 -0.914967487412882
457 -0.91412559598185
458 -0.913286453490408
459 -0.91245004476839
460 -0.911616354763043
461 -0.910785368537852
462 -0.909957071271425
463 -0.909131448256302
464 -0.908308484897891
465 -0.907488166713286
466 -0.906670479330224
467 -0.905855408485983
468 -0.905042940026294
469 -0.904233059904313
470 -0.903425754179528
471 -0.902621009016784
472 -0.901818810685199
473 -0.901019145557206
474 -0.900222000107519
475 -0.899427360912164
476 -0.898635214647502
477 -0.897845548089272
478 -0.897058348111611
479 -0.896273601686146
480 -0.895491295881051
481 -0.894711417860132
482 -0.893933954881903
483 -0.893158894298701
484 -0.892386223555811
485 -0.891615930190562
486 -0.89084800183147
487 -0.890082426197399
488 -0.889319191096684
489 -0.888558284426306
490 -0.887799694171074
491 -0.887043408402804
492 -0.886289415279479
493 -0.885537703044498
494 -0.884788260025833
495 -0.884041074635305
496 -0.88329613536776
497 -0.882553430800321
498 -0.881812949591649
499 -0.881074680481167
500 -0.880338612288344
501 -0.879604733911952
502 -0.878873034329334
503 -0.878143502595726
504 -0.877416127843505
505 -0.876690899281521
506 -0.875967806194384
507 -0.875246837941803
508 -0.87452798395789
509 -0.873811233750505
510 -0.87309657690059
511 -0.872384003061513
};
\addlegendentry{$M_{s}=2$}
\addplot [very thick, green!50!black, dotted]
table {%
1 -60
2 -63.0102999566398
3 -64.7712125471966
4 -66.0205999132796
5 -66.9897000433602
6 -67.7815125038364
7 -68.4509804001426
8 -inf
9 -22.5927084759257
10 -17.4554802119121
11 -14.7210103228566
12 -12.9334508218125
13 -11.6444026997272
14 -10.6582567012089
15 -9.87304210185795
16 -9.22935110225897
17 -8.68980812979243
18 -8.22952822294483
19 -7.83120312628656
20 -7.48236814032019
21 -7.17378774497178
22 -6.89845478815926
23 -6.65094482785225
24 -6.42698543269992
25 -6.22316057028088
26 -6.03670266027816
27 -5.86534312103783
28 -5.70720290536271
29 -5.56071096780733
30 -5.42454261720889
31 -5.29757227007165
32 -5.1788367948702
33 -5.06750675476051
34 -4.96286361603496
35 -4.8642815152475
36 -4.77121254719663
37 -4.68317479911475
38 -4.59974254643724
39 -4.52053816441435
40 -4.44522541248991
41 -4.37350382504946
42 -4.30510399997189
43 -4.2397836204342
44 -4.17732407920397
45 -4.11752760079561
46 -4.06021477724513
47 -4.00522244925694
48 -3.95240187712169
49 -3.90161715586068
50 -3.85274383709955
51 -3.80566772664868
52 -3.76028383200591
53 -3.71649543825655
54 -3.67421329432605
55 -3.63335489439794
56 -3.59384384166512
57 -3.55560928353393
58 -3.51858540902188
59 -3.48271100044401
60 -3.44792903261676
61 -3.41418631376198
62 -3.38143316309796
63 -3.3496231207852
64 -3.31871268647317
65 -3.28866108318618
66 -3.25943004370742
67 -3.23098361698061
68 -3.20328799235772
69 -3.17631133978831
70 -3.15002366427581
71 -3.12439667312524
72 -3.09940365467987
73 -3.07501936739444
74 -3.0512199382235
75 -3.0279827694179
76 -3.00528645292255
77 -2.98311069165629
78 -2.96143622703196
79 -2.94024477214256
80 -2.91951895009947
81 -2.89924223706136
82 -2.87939890953968
83 -2.85997399560767
84 -2.84095322967714
85 -2.82232301053994
86 -2.8040703624001
87 -2.78618289864915
88 -2.76864878815991
89 -2.75145672389536
90 -2.73459589364763
91 -2.71805595273903
92 -2.70182699853211
93 -2.68589954660925
94 -2.67026450849464
95 -2.65491317080238
96 -2.63983717570442
97 -2.6250285026212
98 -2.61047945104582
99 -2.59618262442007
100 -2.58213091498725
101 -2.56831748955298
102 -2.55473577609049
103 -2.54137945113215
104 -2.52824242789335
105 -2.51531884507936
106 -2.5026030563292
107 -2.49008962025435
108 -2.47777329103335
109 -2.46564900952587
110 -2.45371189487309
111 -2.44195723655314
112 -2.43038048686294
113 -2.41897725379986
114 -2.40774329431808
115 -2.39667450793706
116 -2.38576693068035
117 -2.375016729325
118 -2.36442019594292
119 -2.35397374271688
120 -2.34367389701513
121 -2.33351729670948
122 -2.32350068572281
123 -2.31362090979294
124 -2.30387491244054
125 -2.29425973112961
126 -2.28477249360978
127 -2.27541041443051
128 -2.26617079161743
129 -2.25705100350246
130 -2.24804850569899
131 -2.23916082821463
132 -2.23038557269412
133 -2.22172040978554
134 -2.2131630766234
135 -2.20471137442254
136 -2.19636316617718
137 -2.18811637445967
138 -2.17996897931399
139 -2.17191901623907
140 -2.16396457425776
141 -2.15610379406677
142 -2.14833486626402
143 -2.14065602964937
144 -2.13306556959523
145 -2.12556181648373
146 -2.11814314420726
147 -2.11080796872931
148 -2.10355474670289
149 -2.09638197414373
150 -2.08928818515575
151 -2.08227195070647
152 -2.07533187744989
153 -2.06846660659496
154 -2.06167481281724
155 -2.05495520321213
156 -2.04830651628762
157 -2.04172752099488
158 -2.03521701579503
159 -2.02877382776051
160 -2.02239681170949
161 -2.01608484937202
162 -2.00983684858642
163 -2.00365174252474
164 -1.99752848894596
165 -1.99146606947589
166 -1.98546348891249
167 -1.97951977455575
168 -1.97363397556092
169 -1.96780516231427
170 -1.96203242583045
171 -1.95631487717052
172 -1.95065164687982
173 -1.94504188444495
174 -1.9394847577691
175 -1.93397945266484
176 -1.92852517236396
177 -1.92312113704348
178 -1.91776658336725
179 -1.91246076404263
180 -1.9072029473915
181 -1.90199241693525
182 -1.89682847099309
183 -1.8917104222932
184 -1.88663759759624
185 -1.88160933733081
186 -1.87662499524037
187 -1.87168393804117
188 -1.86678554509084
189 -1.86192920806724
190 -1.85711433065724
191 -1.85234032825486
192 -1.84760662766883
193 -1.84291266683876
194 -1.8382578945601
195 -1.83364177021708
196 -1.82906376352383
197 -1.82452335427302
198 -1.82002003209202
199 -1.81555329620608
200 -1.81112265520853
201 -1.8067276268376
202 -1.80236773775971
203 -1.79804252335892
204 -1.79375152753246
205 -1.78949430249205
206 -1.78527040857085
207 -1.78107941403581
208 -1.7769208949053
209 -1.77279443477183
210 -1.76869962462966
211 -1.76463606270718
212 -1.76060335430389
213 -1.75660111163183
214 -1.75262895366132
215 -1.74868650597088
216 -1.74477340060116
217 -1.74088927591288
218 -1.73703377644838
219 -1.73320655279704
220 -1.72940726146416
221 -1.72563556474324
222 -1.72189113059168
223 -1.71817363250974
224 -1.71448274942251
225 -1.71081816556508
226 -1.70717957037055
227 -1.70356665836094
228 -1.69997912904089
229 -1.69641668679404
230 -1.692879040782
231 -1.68936590484593
232 -1.68587699741048
233 -1.68241204139018
234 -1.6789707640982
235 -1.67555289715724
236 -1.67215817641264
237 -1.66878634184776
238 -1.66543713750121
239 -1.6621103113862
240 -1.65880561541183
241 -1.65552280530628
242 -1.65226164054176
243 -1.64902188426137
244 -1.64580330320757
245 -1.64260566765244
246 -1.63942875132956
247 -1.63627233136739
248 -1.63313618822438
249 -1.6300201056254
250 -1.62692387049978
251 -1.62384727292073
252 -1.62079010604611
253 -1.61775216606063
254 -1.6147332521193
255 -1.6117331662922
256 -1.60875171351049
257 -1.60578870151359
258 -1.60284394079767
259 -1.59991724456508
260 -1.59700842867512
261 -1.59411731159576
262 -1.59124371435649
263 -1.58838746050215
264 -1.58554837604783
265 -1.5827262894347
266 -1.57992103148688
267 -1.57713243536908
268 -1.57436033654538
269 -1.57160457273868
270 -1.56886498389114
271 -1.56614141212543
272 -1.56343370170679
273 -1.56074169900594
274 -1.55806525246267
275 -1.5554042125503
276 -1.55275843174085
277 -1.55012776447088
278 -1.54751206710809
279 -1.54491119791866
280 -1.54232501703511
281 -1.539753386425
282 -1.5371961698601
283 -1.53465323288635
284 -1.53212444279429
285 -1.52960966859016
286 -1.52710878096759
287 -1.5246216522798
288 -1.52214815651242
289 -1.5196881692568
290 -1.51724156768389
291 -1.51480823051863
292 -1.51238803801477
293 -1.50998087193034
294 -1.50758661550344
295 -1.50520515342858
296 -1.5028363718335
297 -1.50048015825637
298 -1.49813640162349
299 -1.49580499222741
300 -1.49348582170543
301 -1.49117878301854
302 -1.48888377043075
303 -1.48660067948888
304 -1.48432940700259
305 -1.4820698510249
306 -1.47982191083308
307 -1.47758548690981
308 -1.47536048092474
309 -1.47314679571648
310 -1.47094433527471
311 -1.46875300472286
312 -1.46657271030094
313 -1.46440335934876
314 -1.46224486028945
315 -1.46009712261324
316 -1.4579600568616
317 -1.45583357461161
318 -1.45371758846065
319 -1.45161201201133
320 -1.44951675985673
321 -1.44743174756587
322 -1.44535689166946
323 -1.44329210964588
324 -1.44123731990746
325 -1.43919244178693
326 -1.43715739552416
327 -1.4351321022531
328 -1.433116483989
329 -1.43111046361579
330 -1.42911396487371
331 -1.42712691234715
332 -1.42514923145271
333 -1.42318084842747
334 -1.42122169031741
335 -1.41927168496609
336 -1.41733076100351
337 -1.4153988478351
338 -1.41347587563103
339 -1.4115617753155
340 -1.40965647855647
341 -1.40775991775526
342 -1.40587202603661
343 -1.40399273723872
344 -1.40212198590352
345 -1.40025970726711
346 -1.3984058372503
347 -1.39656031244944
348 -1.39472307012721
349 -1.39289404820373
350 -1.39107318524771
351 -1.3892604204678
352 -1.38745569370406
353 -1.38565894541957
354 -1.38387011669218
355 -1.38208914920635
356 -1.38031598524524
357 -1.37855056768276
358 -1.3767928399759
359 -1.37504274615708
360 -1.37330023082664
361 -1.37156523914555
362 -1.36983771682804
363 -1.36811761013454
364 -1.3664048658646
365 -1.36469943135002
366 -1.36300125444798
367 -1.36131028353437
368 -1.35962646749716
369 -1.35794975572994
370 -1.35628009812549
371 -1.35461744506947
372 -1.35296174743425
373 -1.35131295657275
374 -1.3496710243125
375 -1.34803590294962
376 -1.34640754524307
377 -1.34478590440884
378 -1.34317093411435
379 -1.3415625884728
380 -1.33996082203777
381 -1.33836558979774
382 -1.33677684717077
383 -1.3351945499993
384 -1.33361865454495
385 -1.3320491174834
386 -1.33048589589942
387 -1.32892894728189
388 -1.32737822951895
389 -1.32583370089318
390 -1.32429532007691
391 -1.3227630461275
392 -1.32123683848279
393 -1.31971665695657
394 -1.31820246173412
395 -1.31669421336777
396 -1.31519187277264
397 -1.31369540122231
398 -1.31220476034464
399 -1.31071991211763
400 -1.30924081886527
401 -1.3077674432536
402 -1.30629974828665
403 -1.30483769730257
404 -1.30338125396975
405 -1.30193038228304
406 -1.30048504655993
407 -1.2990452114369
408 -1.29761084186578
409 -1.29618190311009
410 -1.29475836074154
411 -1.29334018063652
412 -1.29192732897263
413 -1.29051977222529
414 -1.28911747716438
415 -1.28772041085092
416 -1.28632854063381
417 -1.28494183414661
418 -1.28356025930436
419 -1.28218378430043
420 -1.28081237760346
421 -1.27944600795427
422 -1.2780846443629
423 -1.27672825610557
424 -1.27537681272183
425 -1.27403028401162
426 -1.27268864003241
427 -1.27135185109645
428 -1.27001988776791
429 -1.26869272086021
430 -1.26737032143328
431 -1.2660526607909
432 -1.26473971047808
433 -1.26343144227844
434 -1.26212782821169
435 -1.26082884053102
436 -1.2595344517207
437 -1.25824463449355
438 -1.25695936178853
439 -1.25567860676832
440 -1.254402342817
441 -1.25313054353765
442 -1.25186318275008
443 -1.25060023448853
444 -1.24934167299943
445 -1.24808747273918
446 -1.24683760837196
447 -1.24559205476756
448 -1.24435078699921
449 -1.24311378034154
450 -1.24188101026843
451 -1.24065245245098
452 -1.23942808275549
453 -1.23820787724144
454 -1.23699181215947
455 -1.23577986394952
456 -1.23457200923879
457 -1.23336822483992
458 -1.23216848774904
459 -1.23097277514393
460 -1.22978106438222
461 -1.22859333299951
462 -1.22740955870764
463 -1.22622971939284
464 -1.22505379311408
465 -1.22388175810126
466 -1.22271359275354
467 -1.22154927563764
468 -1.22038878548619
469 -1.21923210119607
470 -1.21807920182679
471 -1.21693006659889
472 -1.21578467489234
473 -1.21464300624496
474 -1.21350504035092
475 -1.21237075705916
476 -1.2112401363719
477 -1.21011315844315
478 -1.20898980357721
479 -1.20787005222723
480 -1.20675388499379
481 -1.2056412826234
482 -1.20453222600719
483 -1.20342669617944
484 -1.20232467431627
485 -1.20122614173418
486 -1.20013107988886
487 -1.1990394703737
488 -1.19795129491859
489 -1.19686653538859
490 -1.19578517378262
491 -1.19470719223224
492 -1.19363257300036
493 -1.19256129848003
494 -1.1914933511932
495 -1.1904287137895
496 -1.18936736904508
497 -1.18830929986138
498 -1.187254489264
499 -1.18620292040153
500 -1.18515457654441
501 -1.1841094410838
502 -1.18306749753045
503 -1.18202872951361
504 -1.18099312077995
505 -1.17996065519244
506 -1.17893131672934
507 -1.1779050894831
508 -1.17688195765932
509 -1.17586190557574
510 -1.1748449176612
511 -1.17383097845464
};
\addlegendentry{$M_{s}=4$}
\end{axis}

\end{tikzpicture}

%% file: Section/Simulation_Results.tex
\section{Simulation Results}
\label{section:Simulation_results}

This section provides a simulation-based validation of the proposed precoders. The \gls{mrt} is used as a benchmark. Two types of channels are considered. Firstly, the pure LOS channel given in (\ref{eq:channel_model_LOS}) with $\beta=1$. The phases $\phi_m$ are set to zero. Another choice would not have affected performance. Secondly, a general channel is used with i.i.d. Rayleigh distributed components. This often considered model does not assume any spatial correlation between antennas. 
Performance evaluation of the Z3RO precoder based on real-life channel measurements was performed in~\cite{feys2022measurementbased}.

\subsection{Comparison of Local and Global Maxima}

As explained in Corollary~\ref{corollary_LOS}, in pure LOS, all maxima of (\ref{eq:all_real_3rd_problem}) are equally optimal. On the other hand, it is not the case for a general channel. Fig.~\ref{fig:comparison_maxima} shows the array gain achieved by the $M$ maxima described in Theorem~\ref{theorem:maxima}. Indices of antennas are sorted by magnitude of their channel gain, which is also plotted with respect to the right $y$ axis. One can check that the maxima obtained by using the saturated antenna with maximum or minimum channel gain is further from the global maximum. On the other hand, choosing any of the antennas close to the median of the channel gain gives a performance close to the optimal one. In particular, antennas with a relatively low channel gain should be avoided as they require a very high gain to compensate for the distortion of all other antennas. This gives a general and practical guideline to choose which antenna has to be saturated, without having to compare all maxima, as discussed in Remark~\ref{remark:global_optimum}. 

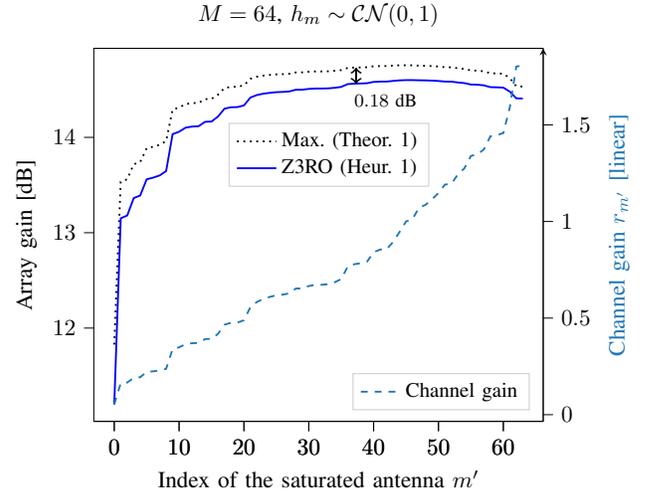
\begin{figure}[t!]
	\centering 
	
	\resizebox{\linewidthmod}{!}{
	\input{Fig/comparison_maxima2.tex}
	}
	
	\caption{Array gain as a function of the saturated antenna $m'$ for the maxima described in Theorem~\ref{theorem:maxima} and the Z3RO precoder in Heuristic~\ref{heuristic_Z3RO} ($\mathcal{M}=\{m'\}$). The right $y$ axis gives the corresponding channel gain $r_{m'}=|h_{m'}|$. Antenna indices are sorted from lowest channel gain to highest.}
	\label{fig:comparison_maxima} 
\end{figure}

\subsection{Comparison of Heuristic Z3RO versus Optimal Precoder}

Fig.~\ref{fig:comparison_maxima} additionally shows the array gain achieved by the heuristic Z3RO precoder described in Heuristic~\ref{heuristic_Z3RO}, using the same antenna as the saturated one. The Z3RO precoder achieves a close-to-optimal performance, with a penalty of about $0.18$ dB.

\subsection{\gls{snr}, \gls{sdr} and \gls{sndr} for Practical PA Models}

\begin{figure}[t!]
	\centering 
	\resizebox{\linewidthmodd}{!}{%
		{\footnotesize \input{Fig/PA_models.tex}}
	} 
	\caption{Two power amplifier models: soft limiter \cite{Tellado2003} and Rapp \cite{Rapp1991EffectsOH}.}
	\label{fig:PA_models} 
	\vspace{-1em}
\end{figure}
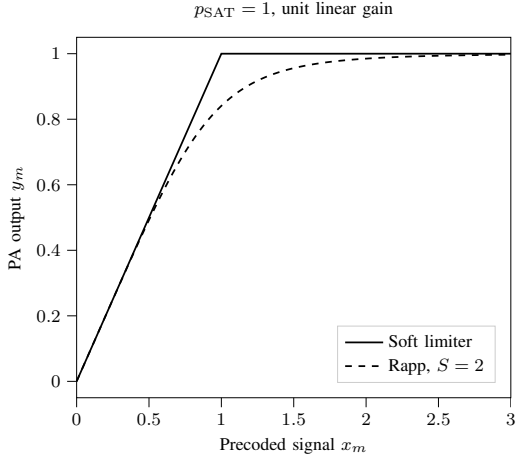


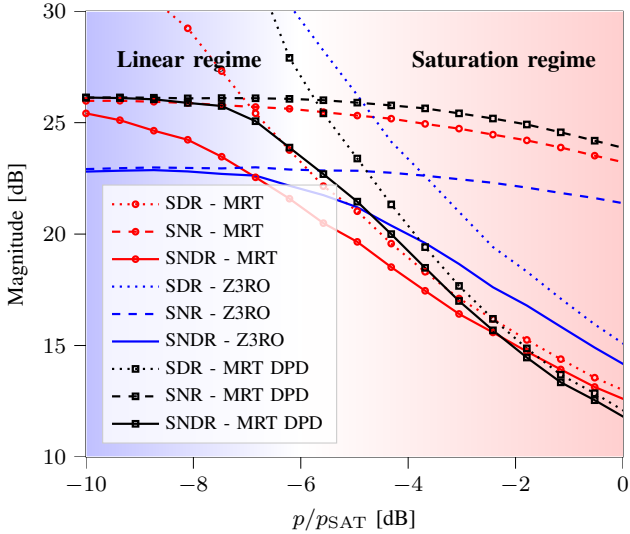
\begin{figure}[t!]
\centering 
\resizebox{\linewidthmod}{!}{%
	{\footnotesize \input{Fig/SDR_SNR_SNDR_soft_lim.tex}}
} 
\caption{\gls{snr}, \gls{sdr} and \gls{sndr} of the \gls{mrt} versus \gls{z3ro} precoders as a function of the back-off (fixed $p_{\mathrm{PA}}=p/M$, varying $p_{\mathrm{SAT}}$). The \gls{z3ro} precoder outperforms \gls{mrt} in the saturation regime, even with a perfect \gls{dpd} implemented.}
\label{fig:SDR_SNR_SNDR_soft_lim} 
\vspace{-1em}
\end{figure}

We now study the performance for two more practical PA models, not limited to a third-order model. They are depicted in Fig.~\ref{fig:PA_models}. The Rapp \gls{pa} model~\cite{Rapp1991EffectsOH} is given by
\begin{align*}
	y_m=\frac{x_m}{\left(1+\left|\frac{x_m}{\sqrt{p_{\mathrm{sat}}}}\right|^{2S}\right)^{\frac{1}{2S}}},
\end{align*}
where $S$ is a smoothness parameter and $p_{\mathrm{sat}}$ is the maximal output power of the PA. As $S \rightarrow +\infty$, we obtain the soft limiter model \cite{Tellado2003}
\begin{align*}
	y_m=\begin{cases}
		x_m &\text{if}\ x_m\leq \sqrt{p_{\mathrm{sat}}}\\
		\sqrt{p_{\mathrm{sat}}} &\text{otherwise}
	\end{cases},
\end{align*}
which can be seen as the input of the antenna $m$ provided that perfect per-antenna \gls{dpd} has been applied before the \gls{pa} \cite{cripps2006rf}. Note that this induces a significant complexity and the \gls{dpd} typically only compensates for weakly nonlinear effects while the clipping due to the finite $p_{\mathrm{sat}}$, \textit{i.e.}, the strongly nonlinear effects, is not compensated.

We compare the \gls{snr}, \gls{sdr} and \gls{sndr} of the \gls{mrt} and the \gls{z3ro} precoders. The Bussgang theorem \cite{bussgang1952crosscorrelation,demir2020bussgang} implies that the received signal can be decomposed as $r=G s + d + v$, where $d$ is the nonlinear distortion, which is uncorrelated with the transmit signal $s$ and noise $v$. The linear gain $G$ can be evaluated as $G=\mathbb{E}(rs^*)/p$. The signal variance is given by $|G|^2p$. Using the fact that $s$, $v$ and $d$ are uncorrelated, the distortion variance is $\mathbb{E}(|d|^2)=\mathbb{E}(|r|^2)-|G|^2p-\sigma_v^2$. The \gls{snr}, \gls{sdr} and \gls{sndr} are thus given by
\begin{align*}
	\text{SNR}&=\frac{|G|^2p}{\sigma_v^2},\ \text{SDR}=\frac{|G|^2p}{\mathbb{E}(|d|^2)},\ \text{SNDR}=\frac{|G|^2p}{\mathbb{E}(|d|^2)+\sigma_v^2},
\end{align*}
where the expectations can be evaluated using the statistics of the transmit symbols $s$. 

For a pure LOS channel, Fig.~\ref{fig:SDR_SNR_SNDR_soft_lim} shows the evolution of the \gls{snr}, \gls{sdr} and \gls{sndr} of the Z3RO versus MRT precoder as a function of the back-off at each antenna $p_{\mathrm{PA}}/p_{\mathrm{SAT}}$, where $p_{\mathrm{PA}}=p/M$. The simulation parameters are: $S=2$, $M=64$, $M_s=4$, $\frac{M\beta p}{\sigma_v^2}=26$ dB, the saturation power $p_{\mathrm{SAT}}$ is varied while $p$ remains fixed. The soft limiter model is used to evaluate the performance of a perfect \gls{dpd}. The signal $s$ has a complex Gaussian distribution. For low values of $p_{\mathrm{PA}}/p_{\mathrm{sat}}$, the \gls{pa} is in the linear regime and the \gls{mrt} achieves an optimal performance. The \gls{z3ro} precoder performs not as well given its reduced array gain. As the ratio $p_{\mathrm{PA}}/p_{\mathrm{sat}}$ increases, the \gls{pa} enters the saturation regime and distortion becomes non-negligible. \gls{mrt} is outperformed by the \gls{z3ro} precoder, which is only limited by distortion orders higher than three. The perfect \gls{dpd} implemented with an MRT precoder can only compensate for weakly nonlinear effects and therefore only improves MRT performance when the PA enters saturation. Close to saturation, the strongly nonlinear effects, which are not compensated by the \gls{dpd}, take over and it is also outperformed by the Z3RO precoder. In conclusion, the advantages of the \gls{z3ro} precoder versus \gls{mrt} (with or without DPD) can be seen in two ways, in the saturation regime:

1) For a same \gls{sndr}, the \gls{z3ro} precoder can work at a larger ratio $p_{\mathrm{PA}}/p_{\mathrm{sat}}$, implying an enhanced energy efficiency. As an example, to achieve a \gls{sndr} of 15 dB the \gls{z3ro} precoder can work with a ratio $p_{\mathrm{PA}}/p_{\mathrm{sat}}$, which is about 1.5 dB higher.

2) For a same $p_{\mathrm{PA}}/p_{\mathrm{sat}}$, the \gls{z3ro} precoder achieves a larger \gls{sndr}, implying an enhanced capacity. As an example, for $p_{\mathrm{PA}}/p_{\mathrm{sat}}=-2$ dB, the \gls{sndr} can be boosted by about 2 dB.

\begin{figure}[t!]
	\centering 
	\resizebox{\linewidthmod}{!}{%
		{\footnotesize \input{Fig/SDR_SNR_SNDR_LOS_varying_p.tex}}
	} 
	\caption{\gls{snr}, \gls{sdr} and \gls{sndr} of the \gls{mrt} versus \gls{z3ro} precoders as a function of the back-off (fixed $p_{\mathrm{SAT}}$, varying $p_{\mathrm{PA}}$). The \gls{z3ro} precoder outperforms \gls{mrt} in the saturation regime, even with a perfect \gls{dpd} implemented.}
	\label{fig:SDR_SNR_SNDR_LOS_varying_p} 
\end{figure}
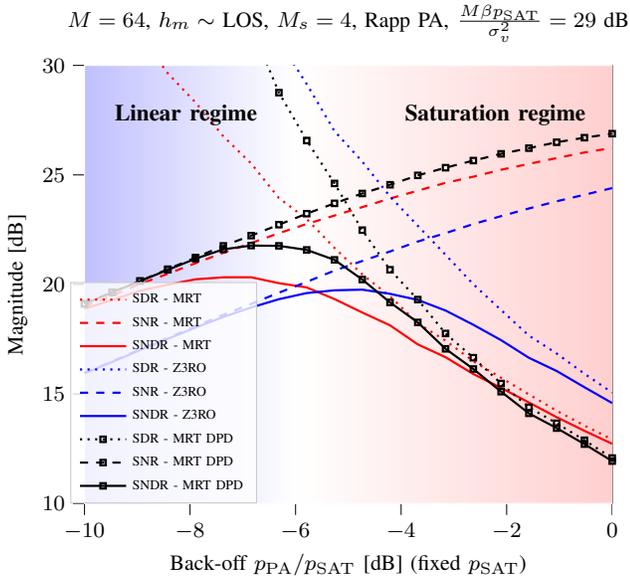

Fig.~\ref{fig:SDR_SNR_SNDR_LOS_varying_p} considers the same simulation parameters as Fig.~\ref{fig:SDR_SNR_SNDR_soft_lim}. The only difference is that $p$ and thus $p_{\mathrm{PA}}=p/M$ is varied, while the saturation power $p_{\mathrm{SAT}}$ remains fixed. As the back-off $p_{\mathrm{PA}}/p_{\mathrm{sat}}$ increases, the SNR of both precoders improves but their SDR degrades. Here again, MRT is only optimal in the linear regime while it is outperformed by the Z3RO precoder in the saturation regime. The DPD improves the SDR performance of MRT but this improvement disappears far in the saturation regime due to clipping.

\subsection{Ergodic Achievable Rate}

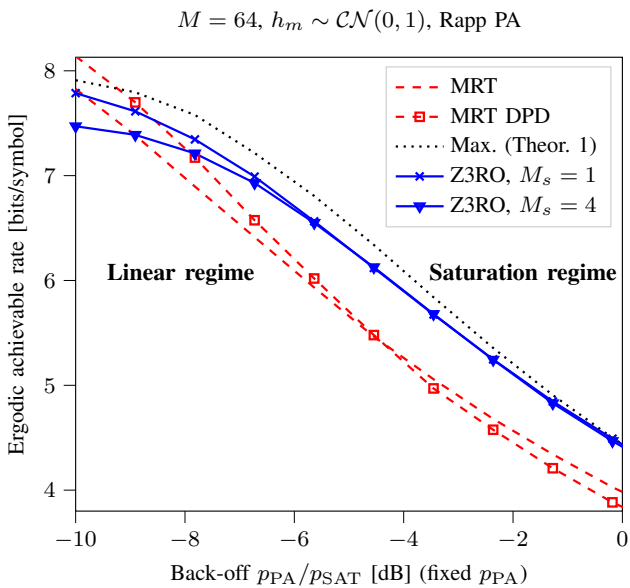
\begin{figure}[t!]
	\centering 
	\resizebox{\linewidthmod}{!}{%
		{\footnotesize \input{Fig/achievable_rate_NLOS_fixed_p.tex}}
	} 
	\caption{Ergodic achievable rates of the MRT, maxima of Theorem~\ref{theorem:maxima} and Z3RO precoder as a function of the back-off (fixed $p_{\mathrm{PA}}$, varying $p_{\mathrm{SAT}}$).}
	\label{fig:achievable_rate_NLOS_fixed_p} 
\end{figure}

The SNDR, computed through the Bussgang theorem, can be converted to an achievable rate $R$ (and thus a lower bound on the capacity), expressed in bits per symbol, by considering the worst case of having noise and distortion Gaussian distributed
\begin{align*}
	R = \log_2(1+\text{SNDR}).
\end{align*}
Taking the expectation of $R$ with respect to the channel statistics, an ergodic achievable rate is obtained. It is shown in Fig.~\ref{fig:achievable_rate_NLOS_fixed_p} for a general channel and for the different studied precoders, as a function of the back-off, fixing $p$ and varying $p_{\mathrm{SAT}}$. As the PA enters saturation, MRT is quickly outperformed by the other precoders. In order, the maximum of Theorem~1 outperforms the Z3RO precoders. The $M_s=1$ Z3RO precoder performs relatively better than the $M_s=4$ version due to its higher array gain. Still, further into saturation, except MRT, their performance converges. This can be intuitively explained by the fact that the performance is limited by the SDR and not the SNR, implying that an array gain penalty becomes less detrimental in that regime.

%% file: Fig/comparison_maxima2.tex
\begin{tikzpicture}

\definecolor{darkgray176}{RGB}{176,176,176}
\definecolor{lightgray204}{RGB}{204,204,204}
\definecolor{steelblue31119180}{RGB}{31,119,180}

\begin{axis}[
legend cell align={left},
legend style={
  fill opacity=0.8,
  draw opacity=1,
  text opacity=1,
  at={(0.3,0.8)},
  anchor=north west,
  draw=lightgray204
},
tick align=outside,
tick pos=left,
x grid style={darkgray176},
xlabel={Index of the saturated antenna $m'$},
xmin=-3.15, xmax=66.15,
xtick style={color=black},
y grid style={darkgray176},
ylabel={Array gain [dB]},
ymin=11.0217898956209, ymax=14.9337618086776,
ytick style={color=black}
]
\addplot [thick, black, dotted]
table {%
0 11.8277504002775
1 13.5352708731257
2 13.5604671770224
3 13.7175356208094
4 13.7412925233039
5 13.8868950596224
6 13.9018594420096
7 13.9216323806003
8 13.9608560846731
9 14.2914801372039
10 14.3128232308411
11 14.348823192979
12 14.3588889614734
13 14.3616226374718
14 14.4006666698685
15 14.4032222109243
16 14.4475366912249
17 14.5168570325664
18 14.5294829923415
19 14.5307213055625
20 14.5452658377746
21 14.6142543914809
22 14.6359253091695
23 14.6477059728834
24 14.6546960387103
25 14.6604890008024
26 14.6625259568423
27 14.6667459932429
28 14.681148163133
29 14.6815367197642
30 14.6881414195965
31 14.6912835082615
32 14.6914674174004
33 14.6928278457419
34 14.6977357968997
35 14.7050093708286
36 14.7281770370605
37 14.7317421018347
38 14.7330817257816
39 14.7353301707856
40 14.7454186238176
41 14.7475848667347
42 14.7478330702395
43 14.7524310793217
44 14.7557500756077
45 14.7559449035386
46 14.7555559367921
47 14.7517588915686
48 14.7503374873204
49 14.7470910025853
50 14.7425748561054
51 14.7353105741942
52 14.7345595145933
53 14.7248566703501
54 14.7178616026786
55 14.7018540761568
56 14.6997073221593
57 14.6926521823068
58 14.6723545813345
59 14.6703793037035
60 14.6669812701731
61 14.6209281695798
62 14.5359233109261
63 14.5328914702699
};
\addlegendentry{\small Max. (Theor. 1)}
\addplot [thick, blue]
table {%
0 11.1996068007599
1 13.1512485917305
2 13.1804601510426
3 13.3628336077322
4 13.3904599319167
5 13.5600295013898
6 13.5774825416837
7 13.6005514092341
8 13.646339361984
9 14.0339329920592
10 14.0590811194105
11 14.1015466163793
12 14.1134316542967
13 14.1166603252554
14 14.1628196718187
15 14.1658440685245
16 14.2183578202589
17 14.3008205905089
18 14.3158922369886
19 14.3173714287875
20 14.3347596250147
21 14.4176901850986
22 14.4439522366172
23 14.4582888397096
24 14.4668192423869
25 14.4739038145168
26 14.476398431509
27 14.4815727940608
28 14.4993026516391
29 14.4997826879355
30 14.5079580955064
31 14.5118586492501
32 14.5120871932012
33 14.5137786570748
34 14.5198941285471
35 14.5290001930581
36 14.5585084698828
37 14.5631551856623
38 14.5649124884143
39 14.5678779971622
40 14.5815467317036
41 14.5846082589399
42 14.5849637338216
43 14.5918336630212
44 14.5978036253822
45 14.6002855677163
46 14.6003620180839
47 14.5984806189163
48 14.5974756649165
49 14.5949751051601
50 14.5912341011829
51 14.5848941254387
52 14.5842245093525
53 14.5754318598829
54 14.5689874185985
55 14.5540953044418
56 14.5520916640821
57 14.5455060443576
58 14.5266184284887
59 14.5247892238545
60 14.5216473946103
61 14.4799130175415
62 14.4093620380165
63 14.4070521727932
};
\addlegendentry{\small Z3RO (Heur. 1)}
\end{axis}

\begin{axis}[
axis y line=right,
legend cell align={left},
legend style={
  fill opacity=0.8,
  draw opacity=1,
  text opacity=1,
  at={(0.97,0.03)},
  anchor=south east,
  draw=lightgray204
},
tick align=outside,
x grid style={darkgray176},
xmin=-3.15, xmax=66.15,
xtick pos=left,
xtick style={color=black},
y grid style={darkgray176},
ylabel=\textcolor{steelblue31119180}{Channel gain $r_{m'}$ [linear]},
title={$M=64$, $h_m \sim \mathcal{CN}(0,1)$},
ymin=-0.0345559425170639, ymax=1.8963685430692,
ytick pos=right,
ytick style={color=black},
yticklabel style={anchor=west}
]
\addplot [thick, dashed, steelblue31119180]
table {%
0 0.053213352282312
1 0.163214244635159
2 0.166584036192566
3 0.190058650274681
4 0.194027648183967
5 0.221305207674275
6 0.224434436025805
7 0.22867466878558
8 0.237461545242343
9 0.339660815011991
10 0.348766107325234
11 0.365154168953996
12 0.369988980521043
13 0.371322381303105
14 0.391389493605283
15 0.392774094613704
16 0.418373855561396
17 0.465848978787418
18 0.475737060338957
19 0.476731316967058
20 0.488759045763364
21 0.55691539588424
22 0.583560836920491
23 0.599545738989174
24 0.609625101938071
25 0.618355907461482
26 0.621513605151698
27 0.62821045950947
28 0.65284310533453
29 0.653550454968265
30 0.665969753203277
31 0.672160417777694
32 0.672528831192371
33 0.675275704952579
34 0.685519064739124
35 0.701780015883239
36 0.766220793467463
37 0.778787679123411
38 0.783789832108259
39 0.792584312016267
40 0.841036073983063
41 0.854558567193803
42 0.856223102802192
43 0.894018071201855
44 0.945834331058256
45 1.00018010802579
46 1.01272565167871
47 1.07158224247022
48 1.08625448686505
49 1.11461750064647
50 1.14731437455377
51 1.19114780539985
52 1.19527203677484
53 1.24401338457075
54 1.27532286800202
55 1.33936428790193
56 1.34735091505926
57 1.37281583734848
58 1.44073995720981
59 1.44700702546712
60 1.45766772805103
61 1.59052061992992
62 1.80159140106682
63 1.80859924826983
};
\addlegendentry{\small Channel gain}
\end{axis}
\draw [<->,thick](4,5.15) -- (4,5.4);
\node at (4.45,4.9) {\footnotesize $0.18$ dB};

\end{tikzpicture}

%% file: Fig/PA_models.tex
\begin{tikzpicture}

\definecolor{darkgray176}{RGB}{176,176,176}
\definecolor{lightgray204}{RGB}{204,204,204}

\begin{axis}[
legend cell align={left},
legend style={
  fill opacity=0.8,
  draw opacity=1,
  text opacity=1,
  at={(0.6,0.2)},
  anchor=north west,
  draw=lightgray204
},
tick align=outside,
tick pos=left,
x grid style={darkgray176},
xlabel={Precoded signal \(\displaystyle x_m\)},
xmin=0, xmax=3,
xtick style={color=black},
title={$p_{\mathrm{SAT}}=1$, unit linear gain},
y grid style={darkgray176},
ylabel={PA output \(\displaystyle y_m\)},
ymin=-0.05, ymax=1.05,
ytick style={color=black}
]
\addplot [thick, black]
table {%
0 0
0.0303030303030303 0.0303030303030303
0.0606060606060606 0.0606060606060606
0.0909090909090909 0.0909090909090909
0.121212121212121 0.121212121212121
0.151515151515152 0.151515151515152
0.181818181818182 0.181818181818182
0.212121212121212 0.212121212121212
0.242424242424242 0.242424242424242
0.272727272727273 0.272727272727273
0.303030303030303 0.303030303030303
0.333333333333333 0.333333333333333
0.363636363636364 0.363636363636364
0.393939393939394 0.393939393939394
0.424242424242424 0.424242424242424
0.454545454545455 0.454545454545455
0.484848484848485 0.484848484848485
0.515151515151515 0.515151515151515
0.545454545454545 0.545454545454545
0.575757575757576 0.575757575757576
0.606060606060606 0.606060606060606
0.636363636363636 0.636363636363636
0.666666666666667 0.666666666666667
0.696969696969697 0.696969696969697
0.727272727272727 0.727272727272727
0.757575757575758 0.757575757575758
0.787878787878788 0.787878787878788
0.818181818181818 0.818181818181818
0.848484848484849 0.848484848484849
0.878787878787879 0.878787878787879
0.909090909090909 0.909090909090909
0.939393939393939 0.939393939393939
0.96969696969697 0.96969696969697
1 1
1.03030303030303 1
1.06060606060606 1
1.09090909090909 1
1.12121212121212 1
1.15151515151515 1
1.18181818181818 1
1.21212121212121 1
1.24242424242424 1
1.27272727272727 1
1.3030303030303 1
1.33333333333333 1
1.36363636363636 1
1.39393939393939 1
1.42424242424242 1
1.45454545454545 1
1.48484848484848 1
1.51515151515152 1
1.54545454545455 1
1.57575757575758 1
1.60606060606061 1
1.63636363636364 1
1.66666666666667 1
1.6969696969697 1
1.72727272727273 1
1.75757575757576 1
1.78787878787879 1
1.81818181818182 1
1.84848484848485 1
1.87878787878788 1
1.90909090909091 1
1.93939393939394 1
1.96969696969697 1
2 1
2.03030303030303 1
2.06060606060606 1
2.09090909090909 1
2.12121212121212 1
2.15151515151515 1
2.18181818181818 1
2.21212121212121 1
2.24242424242424 1
2.27272727272727 1
2.3030303030303 1
2.33333333333333 1
2.36363636363636 1
2.39393939393939 1
2.42424242424242 1
2.45454545454545 1
2.48484848484848 1
2.51515151515152 1
2.54545454545455 1
2.57575757575758 1
2.60606060606061 1
2.63636363636364 1
2.66666666666667 1
2.6969696969697 1
2.72727272727273 1
2.75757575757576 1
2.78787878787879 1
2.81818181818182 1
2.84848484848485 1
2.87878787878788 1
2.90909090909091 1
2.93939393939394 1
2.96969696969697 1
3 1
};
\addlegendentry{Soft limiter}
\addplot [thick, black, dashed]
table {%
0 0
0.0303030303030303 0.0303030239149542
0.0606060606060606 0.0606058561892417
0.0909090909090909 0.0909075386720451
0.121212121212121 0.121205580701156
0.151515151515152 0.151495195339692
0.181818181818182 0.181768542012363
0.212121212121212 0.212013983319263
0.242424242424242 0.242215368530003
0.272727272727273 0.272351361934763
0.303030303030303 0.302394840562033
0.333333333333333 0.33231239229804
0.363636363636364 0.362063951475369
0.393939393939394 0.391602613620419
0.424242424242424 0.420874673138954
0.454545454545455 0.449819926042225
0.484848484848485 0.478372273199711
0.515151515151515 0.506460647187157
0.545454545454545 0.534010267335442
0.575757575757576 0.560944203746401
0.606060606060606 0.587185203612682
0.636363636363636 0.612657705068758
0.666666666666667 0.637289938766844
0.696969696969697 0.661015999452479
0.727272727272727 0.683777762540869
0.757575757575758 0.705526526263034
0.787878787878788 0.726224278586275
0.818181818181818 0.74584451784376
0.848484848484849 0.764372592950344
0.878787878787879 0.781805568197406
0.909090909090909 0.798151653733248
0.939393939393939 0.813429271635689
0.96969696969697 0.827665846241397
1 0.840896415253715
1.03030303030303 0.853162155966903
1.06060606060606 0.864508910823073
1.09090909090909 0.874985781190582
1.12121212121212 0.884643840494158
1.15151515151515 0.893535000019517
1.18181818181818 0.901711044606454
1.21212121212121 0.90922284208869
1.24242424242424 0.91611972016984
1.27272727272727 0.922448997400592
1.3030303030303 0.928255650690048
1.33333333333333 0.933582099830834
1.36363636363636 0.938468089284836
1.39393939393939 0.942950648439418
1.42424242424242 0.947064113253871
1.45454545454545 0.950840194317809
1.48484848484848 0.954308078575153
1.51515151515152 0.957494554149566
1.54545454545455 0.960424149727212
1.57575757575758 0.963119281749294
1.60606060606061 0.965600404215414
1.63636363636364 0.967886157199757
1.66666666666667 0.969993511250837
1.6969696969697 0.971937905706095
1.72727272727273 0.97373337963236
1.75757575757576 0.975392694629946
1.78787878787879 0.976927449138296
1.81818181818182 0.978348184178118
1.84848484848485 0.979664480679676
1.87878787878788 0.980885048696679
1.90909090909091 0.982017808904634
1.93939393939394 0.983069966843572
1.96969696969697 0.984048080397132
2 0.984958121010905
2.03030303030303 0.985805529148499
2.06060606060606 0.986595264468736
2.09090909090909 0.987331851185205
2.12121212121212 0.988019419042784
2.15151515151515 0.988661740316698
2.18181818181818 0.989262263209641
2.21212121212121 0.989824141992506
2.24242424242424 0.990350264205097
2.27272727272727 0.990843275205261
2.3030303030303 0.991305600328551
2.33333333333333 0.991739464895914
2.36363636363636 0.992146912284164
2.39393939393939 0.992529820253011
2.42424242424242 0.992889915703274
2.45454545454545 0.993228788023467
2.48484848484848 0.99354790116606
2.51515151515152 0.993848604580428
2.54545454545455 0.994132143116538
2.57575757575758 0.994399666001742
2.60606060606061 0.994652234982604
2.63636363636364 0.994890831714216
2.66666666666667 0.995116364471027
2.6969696969697 0.995329674245623
2.72727272727273 0.995531540295084
2.75757575757576 0.995722685188486
2.78787878787879 0.995903779403611
2.81818181818182 0.996075445516119
2.84848484848485 0.996238262019981
2.87878787878788 0.996392766814144
2.90909090909091 0.996539460386815
2.93939393939394 0.996678808725679
2.96969696969697 0.996811245979507
3 0.996937176894121
};
\addlegendentry{Rapp, $S=2$}
\end{axis}

\end{tikzpicture}

%% file: Fig/SDR_SNR_SNDR_soft_lim.tex
\begin{tikzpicture}

\definecolor{darkgray176}{RGB}{176,176,176}
\definecolor{green01270}{RGB}{0,127,0}
\definecolor{lightgray204}{RGB}{204,204,204}

\begin{axis}[
legend cell align={left},
legend style={
  fill opacity=0.8,
  draw opacity=1,
  text opacity=1,
  at={(0.03,0.03)},
  anchor=south west,
  draw=lightgray204
},
tick align=outside,
tick pos=left,
x grid style={darkgray176},
xlabel={\(\displaystyle p/p_{\mathrm{SAT}}\) [dB]},
xmin=-10, xmax=0,
xtick style={color=black},
y grid style={darkgray176},
ylabel={Magnitude [dB]},
ymin=10, ymax=30,
ytick style={color=black},
title={$M=64$, $h_m\sim$ LOS, $M_s=4$, Rapp PA, $\frac{M\beta p}{\sigma_v^2}=26$ dB}
]
\fill[draw=white,shade, right color=blue!0, left color=blue!25]
(axis cs:-10,10)--(axis cs:-6,10)--
(axis cs:-6,30)--(axis cs:-10,30)--cycle;
\fill[draw=white,shade, right color=red!25, left color=red!0]
(axis cs:-6,10)--(axis cs:2,10)--
(axis cs:2,30)--(axis cs:-6,30)--cycle;
\addplot [thick, red, dotted, mark=o, mark size=1, mark options={solid}]
table {%
-10 34.5816893988859
-9.36842105263158 32.5129195042852
-8.73684210526316 30.5174679538425
-8.10526315789474 29.2338191448951
-7.47368421052632 27.3117759132804
-6.84210526315789 25.4169303802732
-6.21052631578947 23.7694003664272
-5.57894736842105 22.1510525436401
-4.94736842105263 21.0253271903681
-4.31578947368421 19.5723339501559
-3.68421052631579 18.3083380794855
-3.05263157894737 17.1091744341962
-2.42105263157895 16.1896690706042
-1.78947368421053 15.2509901350176
-1.15789473684211 14.3883543857501
-0.526315789473685 13.556198210847
0.105263157894736 12.892607407656
0.736842105263158 12.1468117343121
1.36842105263158 11.5012519086046
2 10.9952825585972
};
\addlegendentry{\scriptsize SDR - MRT}
\addplot [thick, red, dashed, mark=o, mark size=1, mark options={solid}]
table {%
-10 25.9772829892718
-9.36842105263158 25.9825432958798
-8.73684210526316 25.9375773760679
-8.10526315789474 25.882254487302
-7.47368421052632 25.7853203278207
-6.84210526315789 25.7022212649179
-6.21052631578947 25.6234045038035
-5.57894736842105 25.4729024302207
-4.94736842105263 25.3173119123833
-4.31578947368421 25.179211702905
-3.68421052631579 24.9444227975351
-3.05263157894737 24.7317460067157
-2.42105263157895 24.4632083010364
-1.78947368421053 24.2039313173153
-1.15789473684211 23.8845941024612
-0.526315789473685 23.5188875281066
0.105263157894736 23.1699408302957
0.736842105263158 22.760199152079
1.36842105263158 22.3405310520237
2 21.9093008808243
};
\addlegendentry{\scriptsize SNR - MRT}
\addplot [thick, red, mark=o, mark size=1, mark options={solid}]
table {%
-10 25.4162479618868
-9.36842105263158 25.1107234991459
-8.73684210526316 24.6395635698059
-8.10526315789474 24.2321437235089
-7.47368421052632 23.4715258472407
-6.84210526315789 22.5469336689509
-6.21052631578947 21.5879103837624
-5.57894736842105 20.4915267739569
-4.94736842105263 19.651084796922
-4.31578947368421 18.5172762506906
-3.68421052631579 17.4555539050062
-3.05263157894737 16.4166414437536
-2.42105263157895 15.5871689140458
-1.78947368421053 14.7307334051453
-1.15789473684211 13.9261387773134
-0.526315789473685 13.1388662024901
0.105263157894736 12.5031734846147
0.736842105263158 11.7852023809964
1.36842105263158 11.1572637673312
2 10.6569408256754
};
\addlegendentry{\scriptsize SNDR - MRT}
\addplot [thick, blue, dotted]
table {%
-10 38.7385756505636
-9.36842105263158 38.7171683856068
-8.73684210526316 38.5855624752742
-8.10526315789474 37.3321337926969
-7.47368421052632 35.274148037106
-6.84210526315789 33.471226554394
-6.21052631578947 30.2847957625869
-5.57894736842105 28.2309883403702
-4.94736842105263 26.2353531199491
-4.31578947368421 24.1452280227241
-3.68421052631579 22.5970780493386
-3.05263157894737 20.998996733226
-2.42105263157895 19.418371775471
-1.78947368421053 18.3162457813241
-1.15789473684211 17.1016880893191
-0.526315789473685 15.9490575758975
0.105263157894736 14.9238785803823
0.736842105263158 14.0515884886325
1.36842105263158 13.1806070814622
2 12.4574266755104
};
\addlegendentry{\scriptsize SDR - Z3RO}
\addplot [thick, blue, dashed]
table {%
-10 22.91609130694
-9.36842105263158 22.9573189384525
-8.73684210526316 22.9886978525372
-8.10526315789474 22.9673434005012
-7.47368421052632 22.9471112231412
-6.84210526315789 22.9963526168858
-6.21052631578947 22.896761744779
-5.57894736842105 22.8534815662805
-4.94736842105263 22.8419562767834
-4.31578947368421 22.7520669382261
-3.68421052631579 22.6192216783217
-3.05263157894737 22.4613036677315
-2.42105263157895 22.3011068308805
-1.78947368421053 22.0861510661759
-1.15789473684211 21.8443321529336
-0.526315789473685 21.6152781756488
0.105263157894736 21.3398666760072
0.736842105263158 20.9876908385693
1.36842105263158 20.6368416802692
2 20.257150083302
};
\addlegendentry{\scriptsize SNR - Z3RO}
\addplot [thick, blue]
table {%
-10 22.8039114681556
-9.36842105263158 22.8435306590574
-8.73684210526316 22.8706167453798
-8.10526315789474 22.8112172757461
-7.47368421052632 22.7001249639912
-6.84210526315789 22.6235133909333
-6.21052631578947 22.1688525221334
-5.57894736842105 21.7479186502544
-4.94736842105263 21.2050254067927
-4.31578947368421 20.3827218890512
-3.68421052631579 19.5978357940772
-3.05263157894737 18.6585923841207
-2.42105263157895 17.6145198745394
-1.78947368421053 16.7940715230416
-1.15789473684211 15.8451409648406
-0.526315789473685 14.9067304344783
0.105263157894736 14.0310286843066
0.736842105263158 13.2507970180006
1.36842105263158 12.4631525404707
2 11.7905556663384
};
\addlegendentry{\scriptsize SNDR - Z3RO}
\addplot [thick, black, dotted, mark=square, mark size=1, mark options={solid}]
table {%
-10 156.537026239251
-9.36842105263158 48.1485674880786
-8.73684210526316 43.9847117957292
-8.10526315789474 39.2874639846791
-7.47368421052632 36.7994370677645
-6.84210526315789 31.8200994095186
-6.21052631578947 27.9080056034345
-5.57894736842105 25.4166427177953
-4.94736842105263 23.3905050682437
-4.31578947368421 21.3329077720809
-3.68421052631579 19.414975447473
-3.05263157894737 17.6733438928704
-2.42105263157895 16.1967814431328
-1.78947368421053 14.8707852224946
-1.15789473684211 13.6887306538118
-0.526315789473685 12.8552451153728
0.105263157894736 11.9414822633292
0.736842105263158 11.1884269927887
1.36842105263158 10.4777489520664
2 9.98228159623912
};
\addlegendentry{\scriptsize SDR - MRT DPD}
\addplot [thick, black, dashed, mark=square, mark size=1, mark options={solid}]
table {%
-10 26.1250279736585
-9.36842105263158 26.1321427927262
-8.73684210526316 26.1178862350014
-8.10526315789474 26.0896435219713
-7.47368421052632 26.1028737576019
-6.84210526315789 26.0950638092948
-6.21052631578947 26.0680008177607
-5.57894736842105 26.0064586563617
-4.94736842105263 25.8984424659656
-4.31578947368421 25.7745772408787
-3.68421052631579 25.6405361869688
-3.05263157894737 25.4180826807485
-2.42105263157895 25.185701904534
-1.78947368421053 24.9176612145844
-1.15789473684211 24.5660394329505
-0.526315789473685 24.190759106195
0.105263157894736 23.8068601379359
0.736842105263158 23.3595089194379
1.36842105263158 22.8973475713315
2 22.3912068257377
};
\addlegendentry{\scriptsize SNR - MRT DPD}
\addplot [thick, black, mark=square, mark size=1, mark options={solid}]
table {%
-10 26.1250279736581
-9.36842105263158 26.1049295391439
-8.73684210526316 26.0474855369692
-8.10526315789474 25.8864988485412
-7.47368421052632 25.7478512112126
-6.84210526315789 25.0652196946303
-6.21052631578947 23.8809775636618
-5.57894736842105 22.6912455429886
-4.94736842105263 21.4556011411645
-4.31578947368421 19.9987608501606
-3.68421052631579 18.4861009469262
-3.05263157894737 16.9986035060237
-2.42105263157895 15.6805717512428
-1.78947368421053 14.4610989787222
-1.15789473684211 13.3476268066904
-0.526315789473685 12.5471141963982
0.105263157894736 11.6676510042818
0.736842105263158 10.9326711452936
1.36842105263158 10.2358286664057
2 9.73978234158121
};
\addlegendentry{\scriptsize SNDR - MRT DPD}
\draw (axis cs:-4,27.5) node[
scale=0.5,
anchor=base west,
text=black,
rotate=0.0
]{\LARGE \textbf{Saturation regime}};
\draw (axis cs:-9.5,27.5) node[
scale=0.5,
anchor=base west,
text=black,
rotate=0.0
]{\LARGE \textbf{Linear regime}};
\end{axis}

\end{tikzpicture}

%% file: Fig/SDR_SNR_SNDR_LOS_varying_p.tex
\begin{tikzpicture}

\definecolor{darkgray176}{RGB}{176,176,176}
\definecolor{lightgray204}{RGB}{204,204,204}

\begin{axis}[
legend cell align={left},
legend style={
  fill opacity=0.8,
  draw opacity=1,
  text opacity=1,
  at={(0.15,0.00)},
  anchor=south,
  draw=lightgray204
},
tick align=outside,
tick pos=left,
title={$M=64$, $h_m\sim$ LOS, $M_s=4$, Rapp PA, $\frac{M\beta p_{\mathrm{SAT}}}{\sigma_v^2}=29$ dB},
x grid style={darkgray176},
xlabel={Back-off \(\displaystyle p_{\mathrm{PA}}/p_{\mathrm{SAT}}\) [dB] (fixed $p_{\mathrm{SAT}}$)},
xmin=-10, xmax=0,
xtick style={color=black},
y grid style={darkgray176},
ylabel={Magnitude [dB]},
ymin=10, ymax=30,
ytick style={color=black}
]
\fill[draw=white,shade, right color=blue!0, left color=blue!25]
(axis cs:-10,10)--(axis cs:-6,10)--
(axis cs:-6,30)--(axis cs:-10,30)--cycle;
\fill[draw=white,shade, right color=red!25, left color=red!0]
(axis cs:-6,10)--(axis cs:2,10)--
(axis cs:2,30)--(axis cs:-6,30)--cycle;
\addplot [thick, red, dotted]
table {%
-10 35.0580924558443
-9.47368421052632 32.9022824405204
-8.94736842105263 31.3751345311827
-8.42105263157895 29.6384637727928
-7.89473684210526 28.3027471610151
-7.36842105263158 26.760600796191
-6.84210526315789 25.5215440283135
-6.31578947368421 23.9432923560712
-5.78947368421053 22.9886484082963
-5.26315789473684 21.6653860705104
-4.73684210526316 20.4600070691119
-4.21052631578947 19.4618391931802
-3.68421052631579 18.2255087049091
-3.15789473684211 17.4340676020254
-2.63157894736842 16.4927010247181
-2.10526315789474 15.7366934931159
-1.57894736842105 14.9783971706785
-1.05263157894737 14.2336845190014
-0.526315789473685 13.5558805936655
0 12.9216085271592
};
\addlegendentry{\tiny SDR - MRT}
\addplot [thick, red, dashed]
table {%
-10 18.9984917884299
-9.47368421052632 19.4868206190477
-8.94736842105263 20.0189646780216
-8.42105263157895 20.4857580855903
-7.89473684210526 20.9710476843078
-7.36842105263158 21.4448402853832
-6.84210526315789 21.8907635022493
-6.31578947368421 22.3390882566955
-5.78947368421053 22.7611899085601
-5.26315789473684 23.1442518156752
-4.73684210526316 23.5301343322072
-4.21052631578947 23.9224787469964
-3.68421052631579 24.2820497800649
-3.15789473684211 24.6382256800521
-2.63157894736842 24.9190443848059
-2.10526315789474 25.2364436602215
-1.57894736842105 25.5047543143415
-1.05263157894737 25.7693099311666
-0.526315789473685 26.0263525617197
0 26.2304749805641
};
\addlegendentry{\tiny SNR - MRT}
\addplot [thick, red]
table {%
-10 18.8922002172545
-9.47368421052632 19.2933876126506
-8.94736842105263 19.7122453765505
-8.42105263157895 19.9875993584973
-7.89473684210526 20.2343974221495
-7.36842105263158 20.3253234310465
-6.84210526315789 20.3269895979309
-6.31578947368421 20.0572372141392
-5.78947368421053 19.8631302507698
-5.26315789473684 19.3318726037209
-4.73684210526316 18.7189454807095
-4.21052631578947 18.1327017198914
-3.68421052631579 17.2635730513533
-3.15789473684211 16.6772707678863
-2.63157894736842 15.909694270789
-2.10526315789474 15.2748321751828
-1.57894736842105 14.6097703801486
-1.05263157894737 13.9389692229089
-0.526315789473685 13.3167014674778
0 12.7234770711163
};
\addlegendentry{\tiny SNDR - MRT}
\addplot [thick, blue, dotted]
table {%
-10 39.0885488504755
-9.47368421052632 38.6846737194123
-8.94736842105263 38.5978593842971
-8.42105263157895 38.060072526381
-7.89473684210526 36.5945972609061
-7.36842105263158 35.1426609473046
-6.84210526315789 33.2440371928503
-6.31578947368421 31.1308846656262
-5.78947368421053 29.150001718327
-5.26315789473684 27.0003617451199
-4.73684210526316 25.6109390541854
-4.21052631578947 24.0170059554529
-3.68421052631579 22.6770375278361
-3.15789473684211 21.3527738169348
-2.63157894736842 20.0074223392396
-2.10526315789474 18.8300690306421
-1.57894736842105 17.6645051841935
-1.05263157894737 16.8457872975869
-0.526315789473685 15.9223185906241
0 15.0614542149922
};
\addlegendentry{\tiny SDR - Z3RO}
\addplot [thick, blue, dashed]
table {%
-10 15.9631840242241
-9.47368421052632 16.4878696602447
-8.94736842105263 17.0344194289804
-8.42105263157895 17.5576579355661
-7.89473684210526 18.0901547204871
-7.36842105263158 18.6120199928098
-6.84210526315789 19.1109742709158
-6.31578947368421 19.6224107589825
-5.78947368421053 20.1162328093878
-5.26315789473684 20.6328564378323
-4.73684210526316 21.0695915421869
-4.21052631578947 21.5050047047557
-3.68421052631579 21.9589216259623
-3.15789473684211 22.3412990500467
-2.63157894736842 22.7579101643965
-2.10526315789474 23.1183712343174
-1.57894736842105 23.4676503062484
-1.05263157894737 23.7955744923246
-0.526315789473685 24.0940772790675
0 24.399950331684
};
\addlegendentry{\tiny SNR - Z3RO}
\addplot [thick, blue]
table {%
-10 15.9420883883837
-9.47368421052632 16.461760206923
-8.94736842105263 17.0042247991637
-8.42105263157895 17.5191441913377
-7.89473684210526 18.0293000700935
-7.36842105263158 18.5165345435001
-6.84210526315789 18.9464515384941
-6.31578947368421 19.3259088568655
-5.78947368421053 19.605026273446
-5.26315789473684 19.7309559374983
-4.73684210526316 19.7615872921872
-4.21052631578947 19.5715613521452
-3.68421052631579 19.2928537344967
-3.15789473684211 18.8086714203298
-2.63157894736842 18.1581674016251
-2.10526315789474 17.4548274280873
-1.57894736842105 16.6510343156493
-1.05263157894737 16.0472971707059
-0.526315789473685 15.3064987589369
0 14.5830500700079
};
\addlegendentry{\tiny SNDR - Z3RO}
\addplot [thick, black, dotted, mark=square, mark size=1, mark options={solid}]
table {%
	-10 65.0810965384584
	-9.47368421052632 50.8294423726802
	-8.94736842105263 47.6081235552877
	-8.42105263157895 42.4269675995375
	-7.89473684210526 40.2907071734457
	-7.36842105263158 36.0756191809174
	-6.84210526315789 31.8018805578445
	-6.31578947368421 28.7572238775961
	-5.78947368421053 26.5702076267226
	-5.26315789473684 24.6114476366814
	-4.73684210526316 22.4788423111317
	-4.21052631578947 20.6695815512163
	-3.68421052631579 19.3081003053142
	-3.15789473684211 17.7597551500815
	-2.63157894736842 16.6536413941279
	-2.10526315789474 15.4750354372173
	-1.57894736842105 14.3906541886973
	-1.05263157894737 13.6608061162686
	-0.526315789473685 12.886493846957
	0 12.074609656002
};
\addlegendentry{\tiny SDR - MRT DPD}
\addplot [thick, black, dashed, mark=square, mark size=1, mark options={solid}]
table {%
	-10 19.1149726249031
	-9.47368421052632 19.6457243105155
	-8.94736842105263 20.1573928837632
	-8.42105263157895 20.6959316715819
	-7.89473684210526 21.2288783382864
	-7.36842105263158 21.7603680296157
	-6.84210526315789 22.2294000362771
	-6.31578947368421 22.7287055331832
	-5.78947368421053 23.2306768691508
	-5.26315789473684 23.7015981817569
	-4.73684210526316 24.1524664397162
	-4.21052631578947 24.5509671482431
	-3.68421052631579 24.9789266361423
	-3.15789473684211 25.3296464968267
	-2.63157894736842 25.6547068036356
	-2.10526315789474 25.9624151950942
	-1.57894736842105 26.221781344592
	-1.05263157894737 26.4915862490072
	-0.526315789473685 26.6996448145911
	0 26.8834871632627
};
\addlegendentry{\tiny SNR - MRT DPD}
\addplot [thick, black, mark=square, mark size=1, mark options={solid}]
table {%
	-10 19.1148626821969
	-9.47368421052632 19.6424187339896
	-8.94736842105263 20.1495888168871
	-8.42105263157895 20.6668762529188
	-7.89473684210526 21.1753085187227
	-7.36842105263158 21.6024829728091
	-6.84210526315789 21.7748213306837
	-6.31578947368421 21.7611879347452
	-5.78947368421053 21.5767585143074
	-5.26315789473684 21.1224396506886
	-4.73684210526316 20.2252286228506
	-4.21052631578947 19.1800706088844
	-3.68421052631579 18.2667555057584
	-3.15789473684211 17.0594169687821
	-2.63157894736842 16.1387910904821
	-2.10526315789474 15.1032236764172
	-1.57894736842105 14.1147222617478
	-1.05263157894737 13.4401941791539
	-0.526315789473685 12.7096471811004
	0 11.9334151786408
};
\addlegendentry{\tiny SNDR - MRT DPD}
\draw (axis cs:-4,27.5) node[
scale=0.5,
anchor=base west,
text=black,
rotate=0.0
]{\LARGE \textbf{Saturation regime}};
\draw (axis cs:-9.5,27.5) node[
scale=0.5,
anchor=base west,
text=black,
rotate=0.0
]{\LARGE \textbf{Linear regime}};
\end{axis}

\end{tikzpicture}

%% file: Fig/achievable_rate_NLOS_fixed_p.tex
\begin{tikzpicture}

\definecolor{darkgray176}{RGB}{176,176,176}
\definecolor{lightgray204}{RGB}{204,204,204}

\begin{axis}[
legend cell align={left},
legend style={fill opacity=0.8, draw opacity=1, text opacity=1, draw=lightgray204},
tick align=outside,
tick pos=left,
title={$M=64$, $h_m \sim \mathcal{CN}(0,1)$, Rapp PA},
x grid style={darkgray176},
xlabel={Back-off \(\displaystyle p_{\mathrm{PA}}/p_{\mathrm{SAT}}\) [dB] (fixed $p_{\mathrm{PA}}$)},
xmin=-10, xmax=0,
xtick style={color=black},
y grid style={darkgray176},
ylabel={Ergodic achievable rate [bits/symbol]},
ymin=3.8, ymax=8.1292645806119,
ytick style={color=black}
]
\addplot [thick, red, dashed]
table {%
-10 7.8268829732654
-8.90909090909091 7.38436312335846
-7.81818181818182 6.89599244490852
-6.72727272727273 6.41340737242834
-5.63636363636364 5.93269141005029
-4.54545454545455 5.46471975535714
-3.45454545454546 5.05790262702464
-2.36363636363636 4.67968456552607
-1.27272727272727 4.34160387033028
-0.181818181818184 4.0281094055476
0.909090909090908 3.75891035172013
2 3.52144684819971
};
\addlegendentry{MRT}
\addplot [thick, red, dashed, mark=square, mark size=1.5, mark options={solid}]
table {%
	-10 8.13892319608196
	-8.90909090909091 7.69905135111783
	-7.81818181818182 7.17225626134094
	-6.72727272727273 6.57461724604085
	-5.63636363636364 6.01887329520747
	-4.54545454545455 5.47912900151165
	-3.45454545454546 4.970185042832
	-2.36363636363636 4.57617523027198
	-1.27272727272727 4.20903075100105
	-0.181818181818184 3.8833807175265
	0.909090909090908 3.62297286351466
	2 3.37229419460372
};
\addlegendentry{MRT DPD}
\addplot [thick, black, dotted]
table {%
-10 7.90984468859227
-8.90909090909091 7.79506388983211
-7.81818181818182 7.57868198281756
-6.72727272727273 7.21592420053173
-5.63636363636364 6.8062700593555
-4.54545454545455 6.33930731798387
-3.45454545454546 5.83987429895912
-2.36363636363636 5.35616552065209
-1.27272727272727 4.90691770256586
-0.181818181818184 4.49932665992552
0.909090909090908 4.140083867052
2 3.83057722938584
};
\addlegendentry{Max. (Theor. 1)}
\addplot [thick, blue, mark=x, mark size=2, mark options={solid}]
table {%
-10 7.78705746399298
-8.90909090909091 7.61351217269261
-7.81818181818182 7.34582306099857
-6.72727272727273 6.99068105980481
-5.63636363636364 6.56412727974298
-4.54545454545455 6.11965762688925
-3.45454545454546 5.67559497454368
-2.36363636363636 5.24693822255869
-1.27272727272727 4.84612107393455
-0.181818181818184 4.49077999565162
0.909090909090908 4.1622058443908
2 3.87940854720927
};
\addlegendentry{Z3RO, $M_s=1$}
\addplot [thick, blue, mark=triangle*, mark size=2, mark options={solid,rotate=180}]
table {%
-10 7.47095501166723
-8.90909090909091 7.38929473760188
-7.81818181818182 7.21197096010677
-6.72727272727273 6.92969343214612
-5.63636363636364 6.5495405623423
-4.54545454545455 6.12661439924412
-3.45454545454546 5.68000258379094
-2.36363636363636 5.24405268263465
-1.27272727272727 4.82886661214384
-0.181818181818184 4.46800110536475
0.909090909090908 4.14159237020303
2 3.85126110309711
};
\addlegendentry{Z3RO, $M_s=4$}
\draw (axis cs:-3.6,6) node[
scale=0.5,
anchor=base west,
text=black,
rotate=0.0
]{\LARGE \textbf{Saturation regime}};
\draw (axis cs:-9.5,6) node[
scale=0.5,
anchor=base west,
text=black,
rotate=0.0
]{\LARGE \textbf{Linear regime}};
\end{axis}

\end{tikzpicture}

%% file: Section/Conclusion.tex
\section{Conclusion}
\label{section:conclusion}

In this work, a novel family of linear precoders, the so-called Z3RO family, are proposed for large array-based transmission to allow operating the PAs in an energy efficient operation point closer to saturation. Their implementation has a similar complexity as MRT. In addition, they cancel the third-order distortion at the user location, without requiring knowledge of the PA model and the channel statistics. Their array gain penalty is shown to become negligible for large antenna arrays and especially in the saturation regime. The optimal precoder has no closed-form solution for a general channel. Therefore, the Z3RO precoder was proposed, having a closed-form solution, which is optimal for a pure LOS channel and achieves a good performance for a general channel. Perspectives include the extension to wideband channels, inter-user interference cancellation and \gls{pa} mismatch. The study of the impact of pilot contamination and channel estimation errors is important in future R\&D.

%% file: Section/Appendix.tex
\section{Appendix}
\label{section:appendix}


\subsection{Reasoning behind Conjecture~\ref{conjecture}}\label{appendix:conjecture}

Looking at (\ref{eq:complex_formulation}), if the zero third-order distortion constraint vanishes, it is clear that it is better to have all $g_m$ real and positive to maximize the array gain, as MRT does. However, the zero third-order constraint implies that at least one antenna gain should have an opposite polarity, which degrades the array gain through destructive interference. Several choices are possible for the gains of the antennas with opposite polarity, including complex gains. However, to minimize the array gain degradation, given the transmit power constraint, it is intuitively better to use one strongly saturated (high absolute gain) with opposite polarity rather than several ones being less saturated. Due to the third order, the antenna with opposite polarity ``might just be slightly more amplified" to compensate for the distortion of all others while consuming a minimal transmit power and degrading only slightly the array gain. If a single antenna is used to compensate for the distortion of all others, its polarity should be opposite and thus negative. The problem then becomes fully real-valued. These considerations were confirmed by extensive simulations based on available solvers, for various channel realizations and using different initialization points given the non-concavity of the objective function.

\subsection{Proof of Theorem~\ref{theorem:maxima}}\label{appendix:proof_theorem}

\subsubsection{Critical Points based on First-Order Conditions}

The Lagrangian formulation of problem (\ref{eq:all_real_3rd_problem}) is
\begin{align*}
	L&=\left(\sum_{m=0}^{M-1} g_m r_m\right)^2 - \lambda \left(\sum_{m=0}^{M-1} g_m^2-1\right)-\mu \sum_{m=0}^{M-1} r_m g_m^3.
\end{align*}
The critical points are obtained by setting the derivatives with respect to $g_0,...,g_{M-1}$ to zero while ensuring that the two constraints are satisfied. The derivative with respect to $g_m$ is given by
\begin{align}
	\frac{\partial L}{\partial g_m}&=0\nonumber\\
	\Leftrightarrow 0&=-3r_m\mu g_{m}^2-2\lambda g_{m}+2r_m\sum_{m'=0}^{M-1}g_{m'}r_{m'},\label{eq:FOC_basic}
\end{align}
whose two possible solutions are given by
\begin{align}
	g_m
	&=\frac{-{\lambda}\pm \sqrt{{\lambda^2}+6r_m^2\mu \sum_{m'=0}^{M-1}g_{m'}r_{m'}}}{3\mu r_m}. \label{eq:critical_point_no_lambda}
\end{align}
Now, multiplying (\ref{eq:FOC_basic}) by $g_m$ and summing over $m$ gives
\begin{align}
	0&=-3\mu \sum_{m=0}^{M-1} r_mg_{m}^3-2\lambda \sum_{m=0}^{M-1} g_{m}^2+2\left(\sum_{m'=0}^{M-1}g_{m'}r_{m'}\right)^2\nonumber\\
	\lambda&=\left(\sum_{m'=0}^{M-1}g_{m'}r_{m'}\right)^2, \label{eq:lambda_squared}
\end{align}
where we used the two constraints in (\ref{eq:all_real_3rd_problem}), which should be satisfied for a feasible solution. We thus see that $\lambda$ is given by the squared array gain and is thus positive. As explained before Theorem~\ref{theorem:maxima}, to avoid symmetric equivalent solution, we constrain the solution to lead to a strictly positive array gain. Hence, we can write
\begin{align}
	\sqrt{\lambda}=\sum_{m'=0}^{M-1}g_{m'}r_{m'}>0.
\end{align}
Inserting this result into (\ref{eq:critical_point_no_lambda}) gives
\begin{align}
	g_m	&={-\frac{\lambda}{3\mu r_m}\pm \sqrt{\frac{\lambda^2}{9\mu^2 r_m^2}+\frac{2}{3\mu} \sqrt{\lambda}}}\nonumber\\
	&=\alpha\frac{-1\pm\sqrt{1+r_m^2\xi}}{r_m}, \label{eq:critical_points_simplified}
\end{align}
where $\alpha=\frac{\lambda}{3\mu}$ and $\xi=\frac{6\mu}{\sqrt{\lambda}\lambda}$. They can be related as $\lambda=\frac{4}{(\xi\alpha)^2}$ and $\mu=\frac{\lambda}{3\alpha}=\frac{4}{(\xi\alpha)^2}\frac{1}{3\alpha}$. From (\ref{eq:lambda_squared}), we have $\lambda>0$. For the third-order nonlinear constraint to be valid, $g_m$ should be allowed to have both positive and negative values that can cancel each other. This implies positiveness of $\xi$ and thus $\mu$ and $\alpha$ too. Once the solution $\pm$ for each antenna $m$ is chosen, the values of the constants $\alpha$ and $\xi$ can be found thanks to the two constraints. If the constraints can be satisfied, this results in a given critical point of the problem.

To formalize it, let us define as $\mathcal{M}^{+}$ the set of antenna indices where the solution with sign $+$ is chosen, leading to a positive gain $g_m$. The set of remaining antennas is denoted by $\mathcal{M}^{-}$. These antennas have a negative gain $g_m$. Then, the zero third-order nonlinear constraint implies that
\begin{align*}
	\sum_{m\in \mathcal{M}^+} \frac{(-1+\sqrt{1+r_m^2\xi})^3}{r_m^2}=\sum_{m'\in \mathcal{M}^-} \frac{(1+\sqrt{1+r_{m'}^2\xi})^3}{r_{m'}^2}.
\end{align*}
The value of $\xi$ is the real positive constant that satisfies this equation. It can be found through, \textit{e.g.}, a line search. This equation does not always have a solution, especially if the cardinality of $\mathcal{M}^-$ is too large and/or the channel gains of antennas in $\mathcal{M}^-$ are large. For a given feasible solution $\xi$, the value of $\alpha$ corresponds to the normalization constant, which can be found using (\ref{eq:alpha_def}) as
\begin{align*}
	\alpha=\frac{1}{\sqrt{\sum_{m\in \mathcal{M}^+} \frac{(-1+\sqrt{1+r_m^2\xi})^2}{r_m^2} + \sum_{m'\in \mathcal{M}^-} \frac{(-1-\sqrt{1+r_{m'}^2\xi})^2}{r_{m'}^2}}}.
\end{align*}

\subsubsection{Maxima based on Second-Order Conditions}

The previous section has given the general structure of the critical points of the problem. We now determine which critical points are maxima by analysing the local concavity of the cost function. More specifically, we first show that i) the critical points characterized by a set $\mathcal{M}^{-}$ containing more than one antenna cannot be maxima and ii) a set $\mathcal{M}^{-}$ containing a single element leads to a maximum.

To do this, we study the Hessian of the cost function. We now directly integrate the constraints into the cost function. First of all, the zero third-order nonlinear constraint can be imposed by fixing one of the channel gains. Indeed, without loss of generality, this constraint implies that
\begin{align*}
	r_0g_0^3&=-\sum_{m=1}^{M-1}r_mg_m^3 \Leftrightarrow g_0=-\left(\sum_{m=1}^{M-1}\frac{r_m}{r_0}g_m^3\right)^{1/3}.
\end{align*}
Hence, the optimization can be performed over the remaining variables $g_1,...,g_{M-1}$, while $g_0$ is fixed by the other ones. Similarly, the transmit power constraint can be imposed on the cost function by using the change of variables
\begin{align}
	g_m&=\frac{\tilde{g}_m}{\sqrt{\sum_{m'=0}^{M-1}\tilde{g}_{m'}^2}}, \label{eq:g_m_g_tilde_m}
\end{align}
and performing the optimization with respect to $\tilde{g}_m$. The optimal $g_m$ is then found by normalizing $\tilde{g}_m$ through (\ref{eq:g_m_g_tilde_m}). Hence, problem~(\ref{eq:all_real_3rd_problem}), can be reformulated as
\begin{align}
	\max_{\tilde{g}_1,...,\tilde{g}_{M-1}} f(\tilde{g}_1,...,\tilde{g}_{M-1}) 
	&=\frac{n(\tilde{g}_1,...,\tilde{g}_{M-1})}{d(\tilde{g}_1,...,\tilde{g}_{M-1})}, \label{eq:novel_problem_form}
\end{align}
where
\begin{align*}
	n(\tilde{g}_1,...,\tilde{g}_{M-1})&=\left(\sum_{m=1}^{M-1}r_m \tilde{g}_m + r_0 \tilde{g}_0 \right)^2\\
	d(\tilde{g}_1,...,\tilde{g}_{M-1})&=\sum_{m''=1}\tilde{g}_{m''}^2 + \tilde{g}_0^2\\
	\tilde{g}_0&=-\left(\sum_{m=1}^{M-1}\frac{r_m}{r_0}\tilde{g}_m^3\right)^{1/3}.
\end{align*}
The critical points of (\ref{eq:all_real_3rd_problem}) are properly normalized so that $g_m=\tilde{g}_m$ and are also critical points of (\ref{eq:novel_problem_form}). The other way around, critical points of (\ref{eq:novel_problem_form}) with a proper normalization (\ref{eq:g_m_g_tilde_m}) are also critical points of (\ref{eq:all_real_3rd_problem}). Hence, we can study the Hessian of $f(\tilde{g}_1,...,\tilde{g}_{M-1})$ at each critical point to determine if the critical points found in (\ref{eq:critical_points_simplified}) are maxima or not.

\paragraph{Hessian Computation} At the critical points, some simplifications occurs. The $(\hat{m},\tilde{m})$ element of the Hessian matrix is given by
\begin{align*}
	\frac{\partial^2 f}{\partial \tilde{g}_{\hat{m}}\partial \tilde{g}_{\tilde{m}}}&=\frac{\frac{\partial}{\partial \tilde{g}_{\hat{m}}} (\frac{\partial n}{\partial \tilde{g}_{\tilde{m}}} d - n \frac{\partial d}{\partial \tilde{g}_{\tilde{m}}})d^2-(\frac{\partial n}{\partial \tilde{g}_{\tilde{m}}} d - n \frac{\partial d}{\partial \tilde{g}_{\tilde{m}}}) \frac{\partial}{\partial \tilde{g}_{\hat{m}}}d^2}{d^4}\\
	&=\frac{\partial^2 n}{\partial \tilde{g}_{\tilde{m}}\partial \tilde{g}_{\hat{m}}} - \lambda \frac{\partial^2 d}{\partial \tilde{g}_{\tilde{m}}\partial \tilde{g}_{\hat{m}}},
\end{align*}
where we use the fact that, at critical points, the first order derivative is null, and thus $\frac{\partial n}{\partial \tilde{g}_{\tilde{m}}} d - n \frac{\partial d}{\partial \tilde{g}_{\tilde{m}}}=0$. Moreover, at the critical points (\ref{eq:critical_points_simplified}), $d=1$ and $n=\lambda$. After several derivations, defining $\mat{I}=\mat{I}_{M-1}$, $\tilde{\vect{g}}=(\tilde{g}_1,...,\tilde{g}_{M-1})^T$ and $\vect{r}=(r_0,...,r_{M-1})^T$, we find
\begin{align*}
	\mat{H}&=\frac{\partial^2 n}{\partial \tilde{\vect{g}}\partial \tilde{\vect{g}}^T} - \lambda \frac{\partial^2 d}{\partial \tilde{\vect{g}}\partial \tilde{\vect{g}}^T}\\
	&=2\left(\vect{r} -\frac{1}{\tilde{g}_0^2} \diag(\vect{r}) \tilde{\vect{g}}^2\right)\left(\vect{r} -\frac{1}{\tilde{g}_0^2} \diag(\vect{r}) \tilde{\vect{g}}^2\right)^T\\
	&-2\lambda \left(\mat{I}-\frac{2\diag(\vect{r})\diag(\tilde{\vect{g}})}{r_0\tilde{g}_0} -\frac{\diag(\vect{r})\tilde{\vect{g}}^2(\tilde{\vect{g}}^2)^T\diag(\vect{r})}{r_0^2\tilde{g}_0^4} \right)\\
	&-4  \sqrt{\lambda}   \left(\frac{1}{\tilde{g}_0^3} \diag(\vect{r})\diag(\tilde{\vect{g}})+ \frac{\diag(\vect{r}) \tilde{\vect{g}}^2 (\tilde{\vect{g}}^2)^T\diag(\vect{r}) }{r_0\tilde{g}_0^5}\right).
\end{align*}
At a critical point, we have $\frac{\partial n}{\partial \tilde{\vect{g}}}d=n\frac{\partial d}{\partial \tilde{\vect{g}}}$. Computing these derivatives, using $d=1$ and $n=\lambda$, results in
\begin{align*}
	\left(\vect{r} -\frac{1}{\tilde{g}_0^2} \diag(\vect{r}) \tilde{\vect{g}}^2 \right)&= \sqrt{\lambda}  \left(\tilde{\vect{g}} - \frac{1}{r_0\tilde{g}_0}\diag(\vect{r})\tilde{\vect{g}}^2\right).
\end{align*}
Using this, the Hessian becomes
\begin{align*}
	&\mat{H}=2\lambda\left(\tilde{\vect{g}} - \frac{1}{r_0\tilde{g}_0}\diag(\vect{r})\tilde{\vect{g}}^2\right)\left(\tilde{\vect{g}} - \frac{1}{r_0\tilde{g}_0}\diag(\vect{r})\tilde{\vect{g}}^2\right)^T \\
	&-2\lambda \left(\mat{I}-\frac{2\diag(\vect{r})\diag(\tilde{\vect{g}})}{r_0\tilde{g}_0} -\frac{\diag(\vect{r})\tilde{\vect{g}}^2(\tilde{\vect{g}}^2)^T\diag(\vect{r})}{r_0^2\tilde{g}_0^4} \right)\\
	&-4  \sqrt{\lambda}   \left(\frac{1}{\tilde{g}_0^2} \diag(\vect{r})\diag(\tilde{\vect{g}})+ \frac{\diag(\vect{r}) \tilde{\vect{g}}^2 (\tilde{\vect{g}}^2)^T\diag(\vect{r}) }{r_0\tilde{g}_0^5}\right).
\end{align*}

\paragraph{Non-Maxima among Critical Points}

We now show that the critical points characterized by a set $\mathcal{M}^{-}$ containing more than one antenna cannot be maxima. 
To demonstrate this, we show that the Hessian $\mat{H}$ evaluated at these critical points is not negative semi-definite, which implies that the cost function is not locally concave. Thus, no maximum can be found. Let us consider one such critical point with at least two antenna indices belonging to $\mathcal{M}^{-}$. With a potential reindexing, we can define two of these indices in $\mathcal{M}^{-}$ as antennas $m=0$ and $m=1$ and we further index them so that $r_1 \geq r_0$. From (\ref{eq:critical_points_simplified}), we can find that $g_0<0$ and $g_1<0$ and thus, similarly $\tilde{g}_0=g_0<0$ and $\tilde{g}_1=g_1<0$.

A condition for $\mat{H}$ to be negative semi-definite is that all of its diagonal elements are negative. Let us focus on its first diagonal element
\begin{align*}
	[\mat{H}]_{0,0}&=2\lambda\left(\tilde{g}_1 - \frac{r_1\tilde{g}_1^2}{r_0\tilde{g}_0}\right)^2 -2\lambda \left(1-\frac{2r_1\tilde{g}_1}{r_0\tilde{g}_0} -\frac{r_1^2\tilde{g}_1^4}{r_0^2\tilde{g}_0^4}\right)\\
	&-4  \sqrt{\lambda}   \frac{1}{\tilde{g}_0^5}({\tilde{g}_0^3} r_1\tilde{g}_1+ \frac{r_1^2\tilde{g}_1^4}{r_0}),
\end{align*}
which is the sum of three terms. The first one is clearly positive due to the square and since $\lambda>0$. The third also, given that $\tilde{g}_0<0$, $\tilde{g}_1<0$, $r_0>0$ and $r_1>0$. The second term requires more development. It can be rewritten as
\begin{align*}
	2\lambda \left(-1+\frac{r_1\tilde{g}_1}{r_0\tilde{g}_0} (2+\frac{r_1\tilde{g}_1^3}{r_0\tilde{g}_0^3})\right) \geq 2\lambda \left(-1+2\frac{r_1\tilde{g}_1}{r_0\tilde{g}_0} \right).
\end{align*}
Using (\ref{eq:critical_points_simplified}) and the fact that $r_1 \geq r_0$, we find
\begin{align*}
	\frac{r_1\tilde{g}_1}{r_0\tilde{g}_0}&=\frac{1+\sqrt{1+r_1^2\xi}}{1+\sqrt{1+r_0^2\xi}} \geq 1,
\end{align*}
and thus
\begin{align*}
	2\lambda \left(-1+\frac{r_1\tilde{g}_1}{r_0\tilde{g}_0} (2+\frac{r_1\tilde{g}_1^3}{r_0\tilde{g}_0^3})\right) \geq 2\lambda > 0.
\end{align*}
This implies that $[\mat{H}]_{0,0}>0$ and thus, no critical points with negative gains at more than one antenna can lead to a maximum.

\paragraph{Maxima among Critical Points}\label{appendix:paragraph_maxima}

\begin{figure*}[t!]
	{\footnotesize \begin{align}
			&\left(\mat{I} - \frac{\diag(\vect{r})\diag(\tilde{\vect{g}})}{r_0\tilde{g}_0}\right)^{1/2}\left[-\mat{I}+\left(\mat{I} - \frac{\diag(\vect{r})\diag(\tilde{\vect{g}})}{r_0\tilde{g}_0}\right)^{1/2} \tilde{\vect{g}} \tilde{\vect{g}}^T \left(\mat{I} - \frac{\diag(\vect{r})\diag(\tilde{\vect{g}})}{r_0\tilde{g}_0}\right)^{1/2}\right]\left(\mat{I} - \frac{\diag(\vect{r})\diag(\tilde{\vect{g}})}{r_0\tilde{g}_0}\right)^{1/2}\nonumber\\
			&+\frac{1}{r_0\tilde{g}_0}\diag(\vect{r})\diag(\tilde{\vect{g}})+\frac{1}{r_0^2\tilde{g}_0^4} \diag(\vect{r})\tilde{\vect{g}}^{2}(\tilde{\vect{g}}^{2})^T\diag(\vect{r}). \label{eq:rearrangement}
	\end{align}}
\end{figure*}

We now show that a set $\mathcal{M}^{-}$ containing a single element leads to a maximum. Let us rearrange antenna indices such that the element in $\mathcal{M}^{-}$ corresponds to index $m=0$. From (\ref{eq:critical_points_simplified}), we can find that $g_0=\tilde{g}_0<0$ and $g_m=\tilde{g}_m>0,\ \forall m\neq 0$. We now need to show that the Hessian $\mat{H}$ is semi-definite negative at this critical points. To do this, we use the fact that the sum of negative semi-definite matrices is negative semi-definite. Matrix $\mat{H}$ is the sum of matrices multiplied by either $\lambda$ or $\sqrt{\lambda}$. The group proportional to $4\sqrt{\lambda}$ is
\begin{align*}
-\frac{1}{\tilde{g}_0^2} \diag(\vect{r})\diag(\tilde{\vect{g}})- \frac{1}{r_0\tilde{g}_0^5} \diag(\vect{r}) \tilde{\vect{g}}^2 (\tilde{\vect{g}}^2)^T\diag(\vect{r}) ).
\end{align*}
Multiplying this matrix by $\diag(\vect{r}^{-1/2})\diag(\tilde{\vect{g}}^{-1/2})$ on the left, on the right and by the scalar ${\tilde{g}_0^2}$ does not change its definiteness given that all quantities are positive definite.
\begin{align*}
	- \mat{I}+ \frac{1}{r_0\tilde{g}_0^3} \diag(\vect{r})^{1/2} \tilde{\vect{g}}^{3/2} (\tilde{\vect{g}}^{3/2})^T\diag(\vect{r})^{1/2}.
\end{align*}
This matrix has all eigenvalues equal to $-1$ except a single one equal to $-1-\frac{\vect{r}^T\tilde{\vect{g}}^3}{r_0\tilde{g}_0^3}=-1-\frac{-r_0\tilde{g}_0^3}{r_0\tilde{g}_0^3}=0$. Hence, it is well negative semi-definite. The group proportional to $2\lambda$ is
\begin{align*}
	&\left(\tilde{\vect{g}} - \frac{1}{r_0\tilde{g}_0}\diag(\vect{r})\tilde{\vect{g}}^2\right)\left(\tilde{\vect{g}} - \frac{1}{r_0\tilde{g}_0}\diag(\vect{r})\tilde{\vect{g}}^2\right)^T\\
	&-\left(\mat{I}-\frac{1}{r_0\tilde{g}_0} \diag(\vect{r})\diag(\tilde{\vect{g}})\right)\\
	&+\frac{1}{r_0\tilde{g}_0} \diag(\vect{r})\diag(\tilde{\vect{g}})+\frac{1}{r_0^2\tilde{g}_0^4} \diag(\vect{r})\tilde{\vect{g}}^2(\tilde{\vect{g}}^2)^T\diag(\vect{r}).
\end{align*}
This can can be rewritten as (\ref{eq:rearrangement}). Multiplying (\ref{eq:rearrangement}) by $\diag(\vect{r}^{-1/2})\diag(\tilde{\vect{g}}^{-1/2})$ on the left and on the right, we get
\begin{align*}
	&\left(\mat{I} - \frac{\diag(\vect{r})\diag(\tilde{\vect{g}})}{r_0\tilde{g}_0}\right)^{1/2}\tilde{\vect{g}}\tilde{\vect{g}}^T\left(\mat{I} - \frac{\diag(\vect{r})\diag(\tilde{\vect{g}})}{r_0\tilde{g}_0}\right)^{1/2}\\
	&-\mat{I}+\frac{1}{r_0\tilde{g}_0}(\mat{I}+\frac{1}{r_0\tilde{g}_0^3} \diag(\vect{r}^{1/2})\tilde{\vect{g}}^{3/2}(\tilde{\vect{g}}^{3/2})^T\diag(\vect{r}^{1/2})).
\end{align*}
We can subdivide this matrix in two groups. The first two additive terms give a matrix with all eigenvalues equal to $-1$ except a single one equal to
\begin{align*}
	-1+\|\tilde{\vect{g}}\|^2-\frac{\vect{r}^T\tilde{\vect{g}}^3}{r_0\tilde{g}_0}&=-1+1-\tilde{g}_0^2-\frac{-r_0\tilde{g}_0^3}{r_0\tilde{g}_0}=0.
\end{align*}
The second group can be multiplied by the positive quantity $-r_0\tilde{g}_0$ without changing the definiteness. The remaining matrix has all eigenvalues equal to $-1$ except one equal to $-1-\frac{\vect{r}^T\tilde{\vect{g}}^3}{r_0\tilde{g}_0^3}=0$. As a result the group proportional to $2\lambda$ is well negative semi-definite and so is the Hessian $\mat{H}$ for critical points characterized by a single antenna with negative gain. This concludes the proof of Theorem~\ref{theorem:maxima}.

\subsection{Proof of Corollary~\ref{corollary_LOS}}
\label{appendix:proof_corollary}

Particularizing to the LOS case $r_m=\sqrt{\beta}$, equation (\ref{eq:critical_points_simplified}) implies that the two potential values of $\tilde{g}_m$ are the same for all $m$. We denote these possible values by $\alpha$ and $\delta$. Consider that $M_s$ is the number of coefficients $\tilde{g}_m$ set to $\delta$ and $M-M_s$ are set to $\alpha$, with $M_s>0$ (otherwise the zero distortion constraint cannot be satisfied). For a fixed value of $M_s$, applying the zero distortion constraint gives
\begin{align*}
	\delta&=-\alpha\left(\frac{M-M_s}{M_s}\right)^{1/3}.
\end{align*}
Hence, setting $M_s$ values of $\tilde{g}_m$ to $\delta$ and $M-M_s$ to $\alpha$ give critical points of the Lagrangian. We fix $M/2>M_s$ to avoid symmetrical/equivalent solutions. Among these critical points, we know from Appendix~\ref{appendix:paragraph_maxima} that only the critical points with $M_s=1$ give maxima of the problem. In the LOS case, all of them are globally optimum as they achieve the same array gain. The SNR of the precoder $\text{SNR}_{\mathrm{LOS}}^{\mathrm{Z3RO}}$ in (\ref{eq:SNR_Z3RO}) is found by evaluating the array gain for the above values of the precoder at critical points. The SNR of the MRT can also be evaluated for the LOS channel giving
\begin{align*}
	\text{SNR}^{\mathrm{MRT}}&=\frac{\beta p}{\sigma_v^2} M,\ \text{SNR}_{\mathrm{LOS}}^{\mathrm{Z3RO}}=\frac{\beta p}{\sigma_v^2} M \frac{\left(\zeta^{2/3}-(1-\zeta)^{2/3}\right)^{2}}{\zeta^{1/3}+(1-\zeta)^{1/3}},
\end{align*}
where $\zeta=M_s/M$. For a fixed value of $M_s$, as $M\rightarrow \infty$, $\zeta$ goes to zero and the ratio of SNR goes to one. This concludes the proof of Corollary~\ref{corollary_LOS}.

\subsection{Proof of Proposition~\ref{proposition:Z3RO_large_antenna}}\label{appendix:proof_proposition}

We can define the gain at saturated antennas as
\begin{align*}
		\gamma&=\left(\frac{\sum_{m'=M_s}^{M-1} r_{m'}^4}{\sum_{m''=0}^{M_s-1} r_{m''}^4}\right)^{1/3}\\
		&= \left(\frac{M-M_s}{M_s}\right)^{1/3}\left(\frac{\frac{1}{M-M_s}\sum_{m'=M_s}^{M-1} r_{m'}^4}{\frac{1}{M_s}\sum_{m''=0}^{M_s-1} r_{m''}^4}\right)^{1/3}.
\end{align*}
The SNR of the heuristic precoder is given by
\begin{align*}
	&\text{SNR}^{\mathrm{Z3RO}}=\frac{p}{\sigma_v^2}\frac{\left(-\gamma \sum_{m\in \mathcal{M}} r_m^2  + \sum_{m' \notin \mathcal{M}} r_{m'}^2\right)^2}{\gamma^2 \sum_{m\in \mathcal{M}} r_m^2  + \sum_{m' \notin \mathcal{M}} r_{m'}^2}\\
	&=\frac{p}{\sigma_v^2}M \frac{\left(-\gamma \frac{1}{M}\sum_{m\in \mathcal{M}} r_m  + \frac{1}{M}\sum_{m'' \notin \mathcal{M}} r_{m''}\right)^2}{\gamma^2 \frac{1}{M}\sum_{m\in \mathcal{M}} r_m^2  + \frac{1}{M}\sum_{m''' \notin \mathcal{M}} r_{m'''}^2}.
\end{align*}
We can thus write the following ratio
\begin{align*}
	\frac{\text{SNR}^{\mathrm{Z3RO}}}{\text{SNR}_{\mathrm{LOS}}^{\mathrm{Z3RO}}}=&\frac{1}{\beta}\frac{\left(-\gamma \frac{1}{M}\sum_{m\in \mathcal{M}} r_m^2  + \frac{1}{M}\sum_{m'' \notin \mathcal{M}} r_{m''}^2 \right)^2}{\gamma^2 \frac{1}{M}\sum_{m\in \mathcal{M}} r_m^2  + \frac{1}{M}\sum_{m''' \notin \mathcal{M}} r_{m'''}^2}\\
	&\frac{\zeta^{1/3}+(1-\zeta)^{1/3}}{\left(\zeta^{2/3}-(1-\zeta)^{2/3}\right)^{2}}.
\end{align*}
Given the i.i.d. property of the channel gain, as $M$ and $M_s$ grow large for a fixed value of the ratio $M_s/M$, $\gamma$ converges to $\left(\frac{M-M_s}{M_s}\right)^{1/3}$. We can apply the same results to each average of i.i.d. terms appearing in the above expression. This gives
\begin{align*}
	\frac{\text{SNR}^{\mathrm{Z3RO}}}{\text{SNR}_{\mathrm{LOS}}^{\mathrm{Z3RO}}}\rightarrow&\frac{\left(-\left(\frac{M-M_s}{M_s}\right)^{1/3} \frac{M_s}{M}  + \frac{M-M_s}{M}\right)^2}{\left(\frac{M-M_s}{M_s}\right)^{2/3} \frac{M_s}{M}  + \frac{M-M_s}{M}}\\
	&\frac{\zeta^{1/3}+(1-\zeta)^{1/3}}{\left(\zeta^{2/3}-(1-\zeta)^{2/3}\right)^{2}}.
\end{align*}
Using the change of variable $\zeta=M_s/M$ and rearranging terms, we find that the ratio goes to one. This concludes the proof.